\documentclass[%
 reprint,
 amsmath,amssymb,
 aps,
]{revtex4-1}
\usepackage{amsthm}
\usepackage{comment}
\usepackage{xcolor}
\usepackage{graphicx}
\usepackage{dcolumn}
\usepackage{bm}
\usepackage{wrapfig}
\usepackage{caption}
\usepackage{subcaption}
\usepackage[hidelinks]{hyperref}

\usepackage{physics}

\hypersetup{frenchlinks=true}
\hypersetup{
  colorlinks  = true, 
  urlcolor    = blue, 
  linkcolor   = black, 
  citecolor   = black 
}

\newtheorem{theorem}{Theorem}
\newtheorem{lemma}{Lemma}

\begin{document}

\preprint{APS/123-QED}

\title{Solving graph problems with single-photons and linear optics}

\author{Rawad Mezher}
\email{rawad.mezher@quandela.com}
\author{Ana Filipa Carvalho}
\email{filipa.goncalvescarvalho@quandela.com}
%
\author{Shane Mansfield}%
\email{shane.mansfield@quandela.com}
\affiliation{Quandela SAS, 7 Rue Léonard de Vinci, 91300 Massy, France%
}


\date{\today}

\begin{abstract}
 An important challenge for current and near-term quantum devices is finding useful tasks that can be preformed on them. We first show how to efficiently encode a bounded $n \times n$ matrix $A$ into a linear optical circuit  with $2n$ modes. We then apply this encoding to the case where $A$ is a matrix containing information about a graph $G$. We show that a photonic quantum processor consisting of single-photon sources, a linear optical circuit encoding $A$, and single-photon detectors can  solve a range of graph problems including finding the number of perfect matchings of bipartite graphs, computing permanental polynomials, determining whether two graphs are isomorphic, and  the $k$-densest subgraph problem.
We also propose pre-processing methods to boost the probabilities of observing the relevant detection events and thus improve  performance.
Finally we present both numerical simulations and implementations on Quandela's \textit{Ascella} photonic quantum processor to validate our findings.

\end{abstract}

\maketitle

\section{Introduction}
Quantum computing promises exponential speedups \cite{shor1994algorithms}, breakthroughs in quantum simulation \cite{georgescu2014quantum}, metrology \cite{giovannetti2006quantum,degen2017quantum}, and combinatorial optimization \cite{han2000genetic}, among other advantages. Yet existing and near-term quantum technologies \cite{preskill2018quantum} are extremely prone to errors that hinder performance and prospects of achieving advantages.  At present, 
building a large fault-tolerant quantum computer \cite{lidar2013quantum} remains a formidable technological challenge despite very promising theoretical guarantees \cite{aharonov1997fault}, notably relevant to photonic quantum technologies \cite{rudolph2017optimistic,bartolucci2021fusion,RHG07,kieling2007percolation},  as well as recent experimental advances \cite{arute2019quantum, zhao2022realization, zhong2020quantum,somaschi2016near,coste2022high,postler2022demonstration,marques2022logical}.

Current and near-term quantum devices are in the so-called Noisy Intermediate Scale Quantum (NISQ) regime \cite{preskill2018quantum}. These devices can provide an eventual route to large-scale fault-tolerant architectures, but in the nearer term are also particularly useful for implementing Variational Quantum Algorithms (VQAs) \cite{cerezo2021variational,chabaud2021quantum,perceval} with interesting performances, especially when coupled to error mitigation techniques \cite{endo2021hybrid}.

The focus here is on discrete-variable photonic NISQ devices, essentially composed of single-photon sources \cite{senellart2017high}, linear optical circuits \cite{Reck}, and single-photon detectors \cite{hadfield2009single}. 
Our main contribution is to show that, beyond VQAs,  such platforms can  implement a wide range of promising NISQ algorithms specifically related to solving  linear and graph problems.

We briefly comment on earlier related works in Section \ref{sec:previous} and set up preliminaries in Section \ref{sec:preliminaries}.
In Section \ref{sec:encoding} we find a procedure for encoding matrices, and by extension graphs via their adjacency matrices, into linear optical circuits. We go on to provide an analysis linking the photon detection statistics to properties of the matrices and graphs.
 
In Section \ref{sec:applications} we describe how this method can be used to solve a variety of graph problems \cite{grohe2020graph,graphbook,feige2001dense,merris1981permanental} which are at the core of a of use-cases across diverse fields \cite{fowler2013minimum,cash2000permanental,kasum1981chemical,trinajstic2018chemical,arrazola,kumar1999trawling,fratkin2006motifcut,arora2011computational,raymond2002maximum,grohe2020graph,bonnici2013subgraph}.

In Section \ref{sec:probabilityboosting}, we present pre-processing methods 
which improve the performance of our graph algorithms, thereby improving prospects of achieving \emph{practical} quantum advantages \cite{coyle2021quantum,gonthier2022measurements}.

In Section \ref{sec:numerics}, we perform numerical simulations with the \emph{Perceval} software platform \cite{perceval} to illustrate our encoding and some applications.
Finally in Section \ref{sec:encoding} we implement our methods on the cloud-accessible \textit{Ascella} photonic quantum processor \cite{ascella}, highlighting their interest for near-term technologies.

\section{Previous work}\label{sec:previous}
The encoding procedure we use is similar to the block encoding techniques studied in \cite{low2019hamiltonian,chakraborty2018power}. However, here we study a different set of applications  that can be understood within the Boson Sampling  framework \cite{AA11}
described in terms of linear optical modes and operations rather than qubits and qubit gates.

The Gaussian Boson Sampling (GBS) framework has previously been used to solve graph problems in \cite{bradler1,bradler2,arrazola,schuld}. The main differences between our encoding and that used in \cite{bradler1,bradler2,arrazola,schuld} are first that our setup uses single-photons as input to the linear optical circuit, whereas that of \cite{bradler1,bradler2,arrazola,schuld} uses squeezed states of light. Second, our encoding procedure is more general as it allows for encoding \emph{any} bounded $n \times n$ matrix into a linear optical circuit of $2n$ modes, whereas the encoding in \cite{bradler1,bradler2,arrazola,schuld}, because of the properties of squeezed states, can only encode Hermitian $n \times n$ matrices into 2$n$ mode linear optical setups. We discuss these differences in more detail in Appendix \ref{app:comp}.

\section{Preliminaries}
\label{sec:preliminaries}
We  denote the state of $n$ single-photons arranged in $m$
modes as $|\mathbf{n}\rangle:=|n_1...n_m\rangle$, where $n_i$ is the number of photons in the $ith$ mode, and $\sum_{i=1,..m}n_i=n$. There are $M:= { m+n-1 \choose n}$ distinct (and orthogonal) states of $n$ photons in $m$ modes. These states live in the Hilbert space $\mathcal{H}_{n,m}$ of $n$ photons in $m$ modes, which is isomorphic to the Hilbert space $\mathbb{C}^M$ \cite{AA11,Escartin}. $ \mathsf{U}(m)$ will denote the group of unitary $m \times m $ matrices. A linear optical circuit acting on $m$ modes is represented by a unitary $U \in \mathsf{U}(m)$ \cite{kok2007linear}, and its action on an input state of $n$ photons $|\mathbf{n_{in}}\rangle:=|n_{1,in}...n_{m,in}\rangle$ is given by

\begin{equation}
    \label{eqevolutionLO}
    |\psi\rangle:=\mathcal{U}|\mathbf{n_{in}}\rangle=\sum_{ n_1+...+n_m=n} \gamma_{\mathbf{n}}|\mathbf{n}\rangle,
\end{equation}
where $\mathcal{U} \in \mathsf{U}(M)$ represents the action of the linear optical circuit $U$, and $\gamma_{\mathbf{n}} \in \mathbb{C}$.
Further, we denote 
\begin{equation}
    \label{eqprobabilitiesbs}
    p_U(\mathbf{n}|\mathbf{n_{in}}):=|\gamma_{\mathbf{n}}|^2=\frac{|\mathsf{Per}(U_{\mathbf{n_{in}},\mathbf{n}})|^2}{n_{1,in}!...n_{m,in}!n_1!...n_m!},
\end{equation}
to be the probability of observing  the outcome $\mathbf{n}=(n_1,...,n_m)$ of $n_i$ photons in mode $i$, upon measuring the number of photons in each mode by means of number resolving single-photon detectors \cite{AA11}.  $U_{\mathbf{n_{in}},\mathbf{n}}$ is an $n \times n$ submatrix of $U$ constructed by taking $n_{i,in}$ times the $ith$ column on $U$, and $n_j$ times the $jth$ row of $U$, for $i,j \in \{1,\dots,m\}$ \cite{AA11}. 

 $\mathsf{Per}(.)$ denotes the matrix  permanent \cite{glynn2010permanent}. When there is no ambiguity about the unitary in question we will denote  $p_U(\mathbf{n}|\mathbf{n_{in}})$ as  $p(\mathbf{n}|\mathbf{n_{in}})$ for simplicity.

\section{Encoding}
\label{sec:encoding}
We will now show a method for encoding bounded matrices into linear optical circuits. 
  Let $A \in \mathcal{M}_n(\mathbb{C})$  be an $ n \times n$ matrix with complex entries and of bounded norm, and consider the singular value decomposition \cite{baker2005singular} of $A$
\begin{equation}
    \label{eqsvd}
    A=W\Sigma V^{\dagger},
\end{equation}
where $\Sigma$ is a diagonal matrix of singular values $\sigma_i(A) \geq 0$ of $A$, and $W,V \in \mathsf{U}(n)$. 
Let $s:=\sigma_{max}(A)$ be the largest singular value of $A$, and let
\begin{equation}
    \label{eqencAs}
    A_s:=\frac{1}{s}A=W\left(\frac{1}{\sigma_{max}(A)} \Sigma\right)V^{\dagger}.
\end{equation}
From Eq.(\ref{eqencAs}), it can be seen that $\sigma_{max}(A_s)\leq 1$, and therefore that the spectral norm \cite{horn1990norms} of $A_s$ satisfies
\begin{equation}
    \label{eqencnorm}
    \norm{A_s} \leq 1.  
\end{equation}
With Eq.(\ref{eqencnorm}) in hand, we can now make use of the \emph{unitary dilation theorem} \cite{Halmos}, which shows that when $\norm{A_s} \leq 1$, $A_s$ can be embedded into a larger block matrix
\begin{equation}
  \label{eqencUA}
  U_A:=\begin{pmatrix}
A_s & \sqrt{\mathbb{I}_{n \times n} -A_s(A_s)^{\dagger}} \\
\sqrt{\mathbb{I}_{n \times n} -(A_s)^{\dagger}A_s}  & -(A_s)^{\dagger}\\
\end{pmatrix},
\end{equation}
which is a unitary matrix. Here, $\sqrt{.}$ denotes the matrix square root, and $\mathbb{I}_{n \times n}$ the identity on $\mathsf{U}(n)$ \footnote{Note that since $\norm{A_s} \leq 1$, $\mathbb{I}_{n \times n} -(A_s)^{\dagger}A_s$ is positive semidefinite, and $\sqrt{\mathbb{I}_{n \times n} -(A_s)^{\dagger}A_s}$ is  the unique positive semidefinite matrix which is the square root of $\mathbb{I}_{n \times n} -(A_s)^{\dagger}A_s$. Similarly for $\mathbb{I}_{n \times n} -A_s(A_s)^{\dagger}$ and its square root.}.

Since $U_A \in \mathsf{U}(2n)$, there exists linear optical circuits of $2n$ modes which can implement it \cite{Reck,Clements}. Thus, we have found a way of encoding a (scaled-down version of) $A$ into a linear optical circuit. Note that determining the singular value decomposition  of $A$ can be done in time complexity $O(n^3)$ \cite{vasudevan2017hierarchical,pan1999complexity}.

Furthermore, finding the linear optical circuit for $U_A$ can also be done  in $O(n^2)$ time \cite{Reck,Clements}, thus making our encoding technique efficient. Finally, the choice of rescaling factor $s=\sigma_{max}(A)$ is not unique, as any $s \geq \sigma_{max}(A)$ gives $\norm{A_s} \leq 1$, allowing the application of the unitary dilation theorem \cite{Halmos}. However, choosing $s=\sigma_{max}(A)$ maximizes the output probability corresponding to $\mathsf{Per}(A_s)$, which can be seen from Eq.(\ref{eqprobabilitiesbs}), and the fact that $\mathsf{Per}(A_s)=\frac{1}{s^n}\mathsf{Per}(A)$.

The encoding via Eq.(\ref{eqencUA}) opens up the possibility of estimating $|\mathsf{Per}(A)|$ for any bounded  $A \in \mathcal{M}_{n}(\mathbb{C})$ by using the setup of Figure \ref{fig:BSdevice_graphs} composed of single-photons, linear optical circuits, and  single-photon detectors. Indeed, 
using $U_A$  and
\begin{equation}
\label{eqoutputinputchoice}
    \mathbf{n_{in}}=\mathbf{n}=\underbrace{(1 \dots,1}_{ n \ modes},0 \dots,0),
\end{equation}
 in Eq.(\ref{eqprobabilitiesbs}) gives
\begin{equation}
   \label{eqestimateperm}
    p(\mathbf{n}|\mathbf{n_{in}})=|\mathsf{Per}(A_s)|^2.
\end{equation}
Eq.(\ref{eqestimateperm}) admits a simple interpretation: passing the input $\mathbf{n_{in}}$ of Eq.(\ref{eqoutputinputchoice}) through  the circuit $U_A$ of Eq.(\ref{eqencUA}), then post-selecting on detecting the outcome $\mathbf{n}=\mathbf{n_{in}}$ and using these post-selected samples to estimate  $p(\mathbf{n}|\mathbf{n_{in}})$, allows one to estimate $|\mathsf{Per}(A_s)|$. Since
\begin{equation}
\label{eqrelbtwper}
\mathsf{Per}(A_s)=\mathsf{Per}\left(\frac{1}{s}A\right)=\frac{1}{s^n}\mathsf{Per}(A),
\end{equation}
 then one can also deduce an estimate of $|\mathsf{Per}(A)|$.
 \begin{figure}[h]
    \includegraphics[scale=0.25]{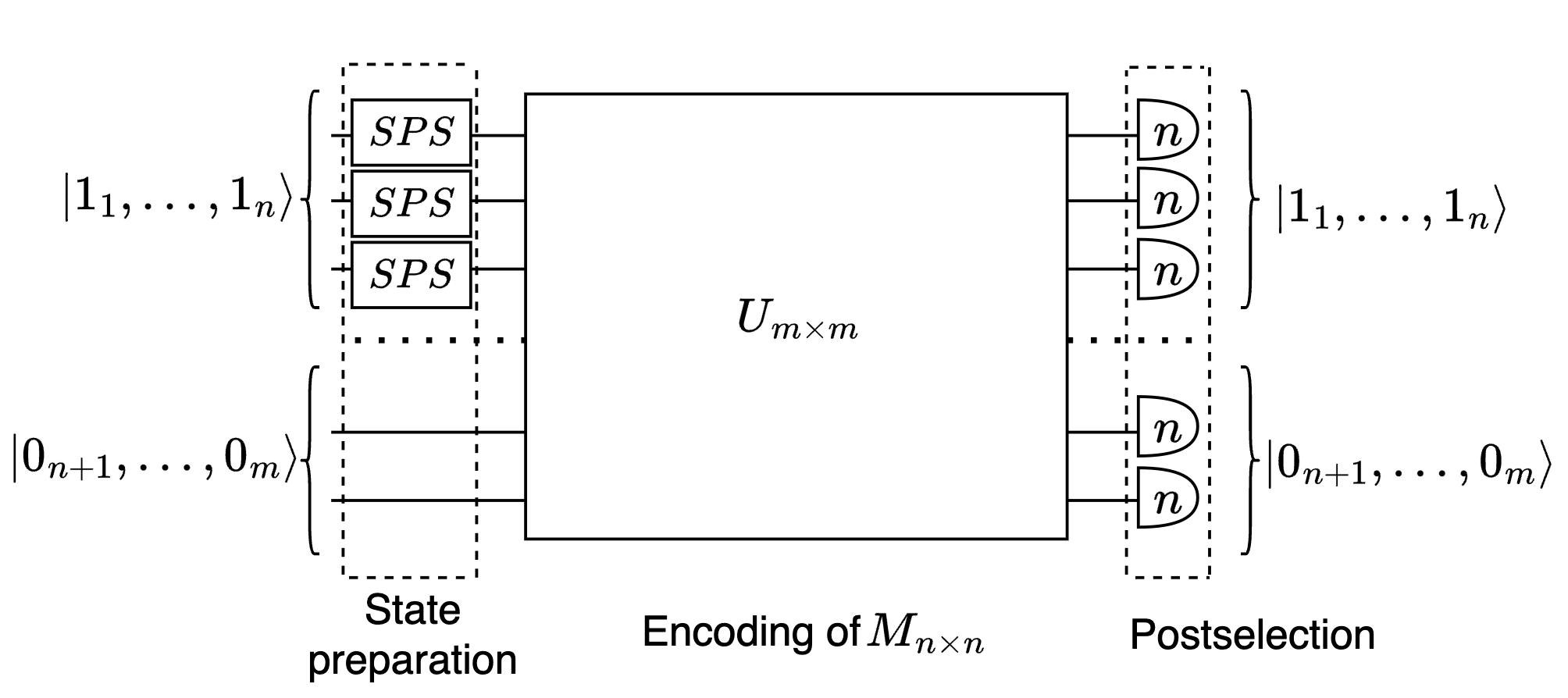}
    \caption{Setup for computing $|\mathsf{Per}(A)|$. First, an input state composed of $n$  single-photons emmitted by $n$ single-photon sources (boxes with $SPS$ label in the figure) is passed through a linear optical circuit  of size $m=2n$ encoding $A$ (the box with $U_{m \times m}$ labelling in the figure). We detect the output state by means of placing $m$ single-photon detectors at the output modes (semi-ellipses with $n$ label), and we post-select on observing $n$ photons in the first $n$ output modes.}
    
    \label{fig:BSdevice_graphs}
\end{figure}
 Furthermore, since we are interested in observing the outcome where at most one photon occupies a mode, then we can use (non-number resolving) threshold detectors, which simplifies the experimental implementation.
\newline \indent

\section{Applications}
\label{sec:applications}
When $A:=(a_{ij})_{i,j \in \{1,..,n\}}$, $a_{ij} \in \{0,1\}$, is the adjacency matrix \cite{graphbook} of a graph $G(V,E)$ (or $G$ for simplicity) with vertex set $V$ composed of $|V|=n$ vertices, and edge set $E$ composed of $|E|=I$ edges, it turns out that computing the permanent of $A$, as well as the permanent of matrices related to $A$, can be extremely useful for a multitude of applications. We will now go on to detail these applications.

\subsection{Computing the number of perfect matchings}
 A perfect matching is a set $E_M \subseteq E$ of independent $\frac{n}{2}$ edges  (no two edges have a common vertex), such that each vertex of $G$ belongs to exactly one edge of $E_M$. When  $G$ is a bipartite graph \cite{graphbook} with its two parts $V_1, V_2 \subset V$ being  of equal size $|V_1|=|V_2|=\frac{n}{2}$, the setup of Figure \ref{fig:BSdevice_graphs} along with Eq.(\ref{eqrelbtwper}) can be used to estimate the number of perfect matchings of $G$ denoted as $\mathsf{pm}(G)$, and given by $\mathsf{pm}(G)=\sqrt{\mathsf{Per}(A)}$ \cite{fuji1997note} \footnote{This holds for an ordering of vertices of $G$ such that we can write $A=\begin{pmatrix}
    0& C \\
    C^T& 0 
\end{pmatrix}$, where $C:=(c_{ij})$ is the biadjacency matrix of $G$, with $c_{ij}=0$ if $i \in V_1$ and $j \in V_2$ are not connected by an edge, and $c_{ij}=1$ otherwise.} \footnote {Note that exactly computing the number of perfect matchings of bipartite graphs is known to be  intractable (more precisely it is $\sharp$$\mathsf{P}$-complete \cite{valiant1979complexity}).}.

\subsection{Computing permanental polynomials}
Our setup can also be used to compute  \emph{permanental polynomials} \cite{merris1981permanental}. These are polynomials, taken here to be over the reals, of the form
\begin{equation}
    \label{eqpermpolynomial}
    P_A(x):=\mathsf{Per}(x\mathbb{I}_{n \times n}-A)=\sum_{i=0,\dots n}c_{i}x^i,
\end{equation}
$x^i$ being the ith power of $x$, and $A$ the adjacency matrix of any graph $G$. The coefficients $\{c_i\}$ are related to the permanents of the subgraphs of $G$ \cite{merris1981permanental}. 

Taking $B_x:=x\mathbb{I}_{n \times n}-A$, we can then compute the coefficients $\{c_i\}$ in Eq.(\ref{eqpermpolynomial}) by performing $n+1$ experiments, where in each experiment we encode $B_{x}$ into a linear optical circuit, and then estimate $\mathsf{Per}(B_{x})$ using the procedure in Figure \ref{fig:BSdevice_graphs} with $A$ replaced by $B_x$. For each experiment $j$,
 we choose a different value $x_j$ of $x$, for $j$ going from 1 to $n+1$. By doing this, we obtain a system of $n+1$ linear equations in $n+1$ unknowns $c_0,\dots, c_n$ . In Appendix \ref{app:permpoly}, we show that almost any random choice of $x_1, \dots, x_{n+1}$ will lead to a solution of this system of linear equations.

 \subsection{Densest subgraph identification}
 
 In the $k$-densest subgraph problem \cite{feige2001dense}, for a given graph $G$ with $n \geq k$ vertices, one must find an induced subgraph (henceforth refered to as subgraph for simplicity) of size $k$ with the maximal density (a $k$-densest subgraph). For a fixed $k$, the densest subgraph is that which has the highest number of edges.
 
 Solving the $k$-densest subgraph problem exactly is $\mathsf{NP}$-Hard \cite{garey1979guide}.

We first give some intuition for why the permanent is a useful tool for identifying dense subgraphs.
 
  Let $\mathcal{S}_n$ be the group of permutations of $\{1,\dots,n\}$, and $A$ the adjacency matrix of $G$. Looking at how the permanent of $A$ is computed
\begin{equation}
    \label{eqpermexpansion}
    \mathsf{Per}(A)=\sum_{\pi \in \mathcal{S}_n}\prod_{i=1 \dots n}a_{i\pi(i)},
\end{equation}
it can be seen that for fixed $n$, the value $\mathsf{Per}(A)$ should increase with increasing the number of non-zero $a_{ij}$, which directly corresponds to increasing the number of edges. 

We make this intuition concrete by proving the following.

\begin{theorem}
\label{thperm}
 For even $n$ and $I$
    \begin{equation}
        \label{equppboundper}
        \mathsf{Per}(A) \leq f(n,I),
    \end{equation}
    where $f(n,I)$ is a function which is monotonically increasing with increasing $I$ for fixed $n$.
    \end{theorem}

Theorem \ref{thperm} is proven in Appendix \ref{app:kdensest}. 
 Theorem \ref{thperm} does not  prove that $\mathsf{Per}(A)$ is a monotonically increasing function of $I$, rather that it is upper bounded by such a function. In Appendix  \ref{app:numerics}, we provide  numerical evidence that for random graphs, $\mathsf{Per}(A)$ is in general a monotonically increasing function of $I$, by plotting the value of $\mathsf{Per}(A)$ versus  $I$, for various values $n$, and for randomly generated graphs.

Taking Eq.(\ref{equppboundper}) together with Eq.(\ref{eqprobabilitiesbs}) we make the  observation that denser subgraphs have a higher probability of being sampled \footnote{This observation was first made in \cite{arrazola}, and applies to our photonic setup as well}. However, at first glance our  setup does not seem very natural for sampling subgraphs. Indeed, subgraphs of $G$  of size $k$ have adjacency matrices of the form $A_{\mathbf{n},\mathbf{n}}$, where $\mathbf{n}=(n_1,...,n_m)$  with $n_i \in \{0,1\}$ and $\sum_{i}n_i=k$ (the same rows and columns are used in constructing the submatrix of $A$) \cite{arrazola}.  However, Eq.(\ref{eqprobabilitiesbs}) shows that submatrices of the form $A_{\mathbf{n_{in}},\mathbf{n}}$ are sampled in  our setup, and these matrices are not in general subgraphs unless $\mathbf{n}=\mathbf{n_{in}}$.

To get around this issue, we  encode a matrix $K$, different to $A$, into our linear optical setup. Consider a set  $\mathcal{S}=\{A_{\mathbf{n_1},\mathbf{n_1}}, \dots,A_{\mathbf{n_J},\mathbf{n_J}} \}$ of $J$ subgraphs of $G$, define $K$ as the block matrix
\begin{equation}
    \label{eqdensesubgraph1}
    K:=\begin{pmatrix} A_{\mathbf{n_1},\mathbf{n_1}} \\ A_{\mathbf{n_2},\mathbf{n_2}} 
    \\ . &&  \mathbf{0}_{kJ \times kJ-k}
    \\ .
    \\ .
    \\ A_{\mathbf{n_J},\mathbf{n_J}} \end{pmatrix}.
\end{equation}
$K$ is an $kJ \times kJ$ matrix, where the $k \times k$ block composed of rows $(j-1)k+1$ to $jk$, and columns 1 to $k$ is the subgraph $A_{\mathbf{n_j}, \mathbf{n_j}}$. Encoding $K$ into a linear optical circuit $U_{{K}}$ of $m:=2kJ$ modes,   and choosing an input of $k$  photons
\begin{equation}\label{eq:inputDSI}
    \mathbf{n_{in}}=\underbrace{(1, \dots,1}_{modes \ 1 \ to \ k},0 \dots,0), 
\end{equation}

then passing this input through the circuit $U_{{K}}$, and post-selecting on observing the outcomes
\begin{equation}\label{eq:outputDSI}
\mathbf{n_{out,j}}:=(0, \dots 0, \underbrace{1, \dots 1}_{modes \ (j-1)k+1 \ to \ jk }, 0 \dots 0),
\end{equation}
for $j \in \{1,\dots,J\}$
allows one to estimate the probabilities 
\begin{multline}
    \label{eqdensesubgraph2}
    p(\mathbf{n_{out,j}}|\mathbf{n_{in}})=\frac{1}{\sigma^{k}_{max}(K)}|\mathsf{Per}(K_{\mathbf{n_{in}},\mathbf{n_{out,j}}})|^2= \\ \frac{1}{\sigma^{k}_{max}(K)}|\mathsf{Per}(A_{\mathbf{n_j},\mathbf{n_j}})|^2.
\end{multline}
As seen previously, the densest subgraph will naturally appear more times in the sampling. 
Note that while $\mathsf{Per}(K)=0$ since it has columns composed entirely of zeros, our procedure relies on sampling from the sub-matrices $A_{\mathbf{n_j},\mathbf{n_j}}$ of $K$ which in general have non-zero permanent and thus non-zero probability of appearing.

The practicality of our setup depends on $J$. For example, if one wants to look at all possible subgraphs of $G$ of size $k$, then $J={n \choose k} \approx n^k$, and therefore $m=2kJ \approx kn^k$, meaning we would need  a linear optical circuit with number of modes exponential in $k$, which is impractical. Nevertheless, we will now show a  useful and practical application for our setup. More precisely, we show how to use our setup with $J=\mathsf{Poly}(n)$ to  improve the solution accuracy of a classical algorithm which approximately solves the $k$-densest subgraph problem \cite{bourgeois2013exact}.

One of the classical algorithms developed in \cite{bourgeois2013exact} approximately solves the $k$-densest subgraph problem by first identifying the $\lceil \rho k \rceil$  vertices, with   $ 0 \leq \rho \leq 1$,  of the densest subgraph of $G$ of size $\lceil \rho k \rceil$, call these vertices $v_1$ to $v_{\lceil \rho k \rceil}$, and  then chooses the remaining $\lfloor (1-\rho)k \rfloor$ vertices arbitrarily.  The identification of vertices  $v_1$ to $v_{\lceil \rho k \rceil}$ is done through an algorithm which \emph{exactly} solves the $\lceil \rho k \rceil$-densest   subgraph problem. The runtime of this algorithm is  $O(c^n)$ and thus exponential in $n$, where $c>1$  generally depends on the ratio $\frac{\rho k}{n}$ \cite{bourgeois2013exact}.\newline \indent
Our approach is to replace the arbitrary choice of the remaining $\lfloor (1-\rho)k \rfloor$ vertices from \cite{bourgeois2013exact}  with the following algorithm. First, identify all subgraphs of size $k$ with their first $\lceil \rho k \rceil$ vertices being $v_1$ to $v_{\lceil \rho k \rceil}$. Then, encode these into our setup  (see Eq.(\ref{eqdensesubgraph1}) to (\ref{eqdensesubgraph2})). The number of these subgraphs is 
\begin{equation}
\label{eqvalofM}
J={n-\lceil \rho k \rceil \choose \lfloor (1-\rho)k \rfloor } \leq n^{(1-\rho)k}.
\end{equation}
Choosing $\rho=1-O(\frac{1}{k})$, and substituting into Eq.(\ref{eqvalofM}) gives $$J \approx n^{O(1)} = \mathsf{Poly}(n).$$
Thus, when a  majority  of vertices of the densest subgraph have been determined classically, our encoding can be  used to identify the remaining vertices, by using linear optical circuits acting on a $\mathsf{Poly}(n)$ number of modes.

Depending on how accurately we estimate the probabilities in Eq.(\ref{eqdensesubgraph2}), we can in principle boost the accuracy of the approximate solution of \cite{bourgeois2013exact}.
\newline

\subsection{Graph Isomorphism}

Given two (unweighted, undirected) graphs $G_1(V_1,E_1)$ and $G_2(V_2,E_2)$ with $|V_1|=|V_2|=n$, and with respective adjacency matrices $A$ and $B$, $G_1$ is isomorphic to $G_2$ iff $B=P_{\pi}AP^T_{\pi}$, for some $P_{\pi} \in \mathcal{P}_n$, the group of $n \times n$ permutation matrices. We now explore the graph isomorphism problem (GI): the problem of determining whether two given graphs are isomorphic 

\footnote{GI is believed to lie in the complexity class $\mathsf{NP}$-intermediate, with the best classical algorithm for determining whether two graphs are isomorphic running in quasipolynomial time \cite{babai2016graph}}. GI has previously been investigated in  the framework of quantum walks \cite{QW1,QW2}, as well as in Gaussian Boson Sampling \cite{bradler2}. Here, we show how to use our photonic setup to solve GI.

More concretely, let $l \in \{1,\dots,n\}$, 
and $\mathbf{t}:=\{t_1,\dots,t_l\}, \mathbf{s}:=\{s_1,\dots,s_l\}$ with $s_i,t_i \in \{1, \dots, n\}$, and $s_{i+1} \geq s_i$, $t_{i+1} \geq t_i$, for all $i$. Let $B_{\mathbf{t}, \mathbf{s}}$ be an $l \times l$ submatrix of $B$ constructed first by constructing an $l \times n$ matrix $B_{\mathbf{s}}$ such that the $ith$ row of $B_{ \mathbf{s}}$ is the $s_ith$ row of $B$, and then constructing  $B_{\mathbf{t}, \mathbf{s}}$ such that its $jth$ column is the $t_jth$ column of $B_{\mathbf{s}}$. In Appendix \ref{app:GI} we prove the following theorem.
\begin{theorem}
\label{th1}
Let $G_1$ and $G_2$ be two unweighted, undirected, isospectral (having the same eigenvalues) graphs with $n$ vertices, and with no self loops. Let   $A$ and $B$ be the respective adjacency matrices of $G_1$ and $G_2$. The following two statements are equivalent:  \\ 

1) There exists a fixed bijection $\pi:\{1,\dots,n\} \to \{1,\dots,n\} $ such that for all $l$, $\mathbf{s}$, $\mathbf{t}$, the following is satisfied $$\mathsf{Per}(A_{\mathbf{\pi(t)},\mathbf{\pi(s)}})=\mathsf{Per}(B_{\mathbf{t},\mathbf{s}}),$$
with $\mathbf{\pi(s)}=\{\pi(s_1),\dots, \pi(s_l)\}$, $\mathbf{\pi(t)}=\{\pi(t_1),\dots, \pi(t_l)\}$.

\bigskip
2) $G_1$ is isomorphic to $G_2$.
\end{theorem}

Practically,  Theorem \ref{th1} implies that a protocol consisting of  encoding $A$ and $B$ into linear optical circuits $U_{A}$ and $U_{B}$, and examining  the output probability distributions resulting from passing $l$ single-photons through $U_{A}$
and $U_{B}$, for variable $l$ ranging from 1 to $n$ and for all possible ${n+l-1 \choose l}$  arrangements of $l$ input photons in the first $n$ modes, is  necessary and sufficient  for $G_1$ and $G_2$ to be isomorphic.

Theorem \ref{th1} is an interesting theoretical observation, but its  utility as a method for solving GI is clearly limited by the number of required experimental rounds, $\sum_{l=1,\dots,n}{n+l-1 \choose l}$ , which scales exponentially in $n$. However, by using the fact that our setup naturally computes permanents, we can import powerful permanent-related tools from the field of graph theory to distinguish non-isomorphic graphs \cite{merris1981permanental,wu2022characterizing,liu2017bivariate}. For example, one of these tools, which we use in our numerical simulations in Section \ref{sec:numerics}, is the  Laplacian permanental polynomial \cite{merris1981permanental}, defined here  over the reals, which for a  graph $G$ with Laplacian $L(G)$ has the form
\begin{equation}
\label{eqlappermpoly}
    P_L(x):=\mathsf{Per}(x\mathbb{I}_{n \times n}-L(G)).
    \end{equation}
    
Laplacian permanental polynomials are  particularly useful for GI. It is known that equality of the Laplacian permanental polynomials of $G_1$ and $G_2$ is a necessary condition for these graphs to be isomorphic \cite{merris1981permanental}. Furthermore, this equality is known to be a necessary and sufficient condition within many families of graphs \cite{wu2022characterizing,liu2019signless}, although families are also known for which sufficiency does not hold \cite{merris1991almost}.
Other polynomials based on permanents are also studied \cite{liu2017bivariate,liu2019signless}.  All of these polynomials can be computed within our setup, similarly to how one would compute the polynomial of Eq.(\ref{eqpermpolynomial}).

    \section{Sample Complexities}
    \label{sec:samplecomplexities}
    At this point we comment on the distinction between \emph{estimating} $|\mathsf{Per}(A)|$ and (exactly) computing $|\mathsf{Per}(A)|$ for some $A \in \mathcal{M}_n(\mathbb{C})$. When running experiments using our setup, one obtains an estimate of $|\mathsf{Per}(A)|$, by estimating $p(\mathbf{n}|\mathbf{n_{in}})$ from samples obtained from many runs of an experiment. 
    
    With this in mind, one can use Hoeffding's inequality \cite{hoeffding1994probability} to estimate $p(\mathbf{n}|\mathbf{n_{in}})$ (and consequently $|\mathsf{Per}(A)|$) to within an additive error $\frac{1}{\kappa}$ by performing $O(\kappa^2)$ runs, with $\kappa \in \mathbb{R^{+*}}$. In practice, one usually aims at performing an efficient number of runs, that is $\kappa =\mathsf{Poly}(n)$. At this point, it becomes clear that estimating permanents using our devices will not give a superpolynomial quantum-over-classical advantage, as for example the classical Gurvits algorithm \cite{gurvits2005complexity,aaronson2012generalizing} can estimate permanents to within $\frac{1}{\mathsf{Poly}(n)}$ additive error in $\mathsf{Poly}(n)$-time. However, our techniques can still potentially lead to \emph{practical} advantages \cite{coyle2021quantum,gonthier2022measurements} over their classical counterparts for specific examples and in specific applications.

   \section{Probability Boosting}
   \label{sec:probabilityboosting}
    We strengthen the case for practical advantage by demonstrating two techniques which allow for a better approximation of $\mathsf{Per}(A)$ using less samples. These techniques boost the probabilities of seeing the most relevant outcomes. They rely on modifying the matrix $A$, then encoding these modified versions in our setup. However, care must be taken so that the modifications allow us to  efficiently recover back the value of $\mathsf{Per}(A)$.

    Let $\mathbf{A}_{i}$ denote the $ith$ row of $A$, and $c \in \{1,\dots,n\}$ be a fixed row number. Let $A_w$ be a matrix, its $ith$ row $\mathbf{A_w}_{i}$ is given by:  $\mathbf{A_w}_{i}=\mathbf{A}_i$ for all $i \neq c$, and $\mathbf{A_w}_{c}=w\mathbf{A}_c$, with $w \in \mathbb{R}^{+*}$ .
    Our first technique for boosting is inspired by the observation following from Eq.(\ref{eqpermexpansion}) that
    
    \begin{equation}
        \label{eqrelAwA}
        \mathsf{Per}(A_w)=w\mathsf{Per}(A).
    \end{equation}
    Thus, when $w>1$, this modification boosts the value of the permanent. However, in order to boost the probability of appearance of desired outputs using this technique, the ratio of the largest singular values $\sigma_{max}(A)$ and $\sigma_{max}(A_w)$ of  $A$ and $A_w$ must be carefully considered (see Eq.(\ref{eqrelbtwper})).
     
    In  Appendix \ref{app:boosting}, we show that
    \begin{equation}
    \label{eqboostmaintext}
    \frac{\sigma_{max}(A)}{\sigma_{max}(A_w)} > \frac{1}{w^{\frac{1}{n}}},
    \end{equation}
    is a necessary condition for boosting to occur using this technique. We also find examples of graphs $G$ where this condition is satisfied.
   
    For a fixed $n>1$, in the limit of large $w$, we find  that $\sigma_{max}(A_w) \approx O(w)$ (see Lemma \ref{lemboundsAw}), meaning $\frac{\sigma_{max}(A)}{\sigma_{max}(A_w)} \approx O(\frac{1}{w})$, indicating that the condition of Eq.(\ref{eqboostmaintext}) is violated. This means that beyond some value $w_0$ of $w$, depending on $A$, boosting no longer occurs.
    \newline \indent
  
   The second technique for probability boosting we develop  takes inspiration from the study of permanental polynomials \cite{merris1981permanental}. Consider the matrix
    \begin{equation}
    \label{eqAeps}
    \tilde{A}_{\varepsilon}=A+\varepsilon \mathbb{I}_{n \times n},
    \end{equation}
    with $\varepsilon \in \mathbb{R}^+$.
    Using the expansion formula for the permanent of a sum $A+\varepsilon \mathbb{I}_{n \times n}$ of two matrices \cite{krauter1987theorem}, we obtain
    \begin{equation}
        \label{eqAeps2}
        \mathsf{Per}(\tilde{A}_{\varepsilon})=\mathsf{Per}(A)+\sum_{i=1,..,n}c_i\varepsilon^{i},
    \end{equation}
   where, as in the case of the permanental polynomial, $c_i$ is a sum of permanents of submatrices of $A$ of size $n-i \times n-i$ \cite{krauter1987theorem}. If $A$ is a matrix with non-negative entries, then $c_i \geq 0$, and therefore
   $$\mathsf{Per}(\tilde{A}_{\varepsilon}) \geq \mathsf{Per}(A).$$
   Here again, the value of the permanent is boosted, and one can recover the value of $\mathsf{Per}(A)$ by computing $\mathsf{Per}(\tilde{A}_{\varepsilon})$ for $n+1$ different values of $\varepsilon$, then solving the system of linear equations in $n+1$ unknowns to determine the set of values $\{\mathsf{Per}(A),\{c_i\}\}$. As with the previous technique, the boosting provided by this method ceases after a certain value $\varepsilon_0$ of $\varepsilon$ for a fixed $n$, as shown in  Appendix \ref{app:boosting}. 
\section{Numerical simulations}
\label{sec:numerics}
In this section, we highlight some of the numerical simulations  performed to test our encoding as well as our applications. All  simulations were performed using the \emph{Perceval} software platform \cite{perceval}. Our code, as well as a full description of how to use it is available at \footnote{\href{https://github.com/Quandela/matrix-encoding-problems}{https://github.com/Quandela/matrix-encoding-problems}.}.

 We performed simulations for estimating the permanent of a matrix $A$ by encoding it into a linear optical circuit, and post-selecting as in Figure \ref{fig:BSdevice_graphs}. An example is provided in Table \ref{table:permEst}. We construct random graphs of six vertices with various edge probabilities $p \in [0,1]$, where $p$ represents the probability that vertices $i$ and $j$ are connected by an edge. For each $p$, we construct four 
random graphs with respective adjacency matrices $A_1, \dots, A_4$, and compute an estimate $\mathsf{E}(\mathsf{Per}(A_r))$ of $\mathsf{Per}(A_r)$ for each $r \in \{1, 2, 3, 4\}$. This estimate is computed by using 500 post-selected samples. We then compute the mean estimate $\mu_{estimate}:=\sum_{r=1, \dots, 4} \frac{\mathsf{E}(\mathsf{Per}(A_r))}{4}$. Table \ref{table:permEst} shows $\mu_{estimate}$ and the mean exact value  $\mu_{exact}:=\sum_{r=1, \dots, 4} \frac{\mathsf{Per}(A_r)}{4}$ with respect to $p$. As can be seen in  Table \ref{table:permEst}, a close agreement is observed between exact and estimated values.

\begin{table}[h]
\centering
  \begin{tabular}{ | c | c | c |} \hline
      $p$ & $\mu_{exact}$ & $\mu_{estimate}$\\ \hline\hline
     $0.70$ & $43.25$ & $44.34$\\ \hline
     $0.78$ & $83.50$ & $82.97$\\ \hline
     $0.86$ & $109.25$ & $109.36$  \\ \hline
     $0.94$ & $155.25$ & $156.97$ \\ \hline
     $1.00$ & $265.00$ & $265.66$ \\ \hline
  \end{tabular}
\caption{Mean value of estimation and calculation of permanent for random graphs of 6 vertices. For each edge probability 4 graphs were generated, and the mean estimated value of the permanent is computed by taking the average of the estimated values of the permanent of these four graphs. These estimates are obtained from $500$ post-selected samples}
\label{table:permEst}
\end{table}

For dense subgraph identification we  wrote code which, given access to a subset $\mathcal{S}$ (of size less than $k$) of  vertices of the densest subgraph, first constructs all possible subgraphs of size $k$ containing 
all vertices from $\mathcal{S}$,  then encodes these subgraphs into a single linear optical circuit (see Eq.(\ref{eqdensesubgraph1})-(\ref{eqdensesubgraph2})), and samples outputs from this circuit.
To test our code and our technique, we considered the graph of Figure \ref{fig:randomDSI}.
 \begin{figure}[h]
    \centering
    \includegraphics[scale=0.5]{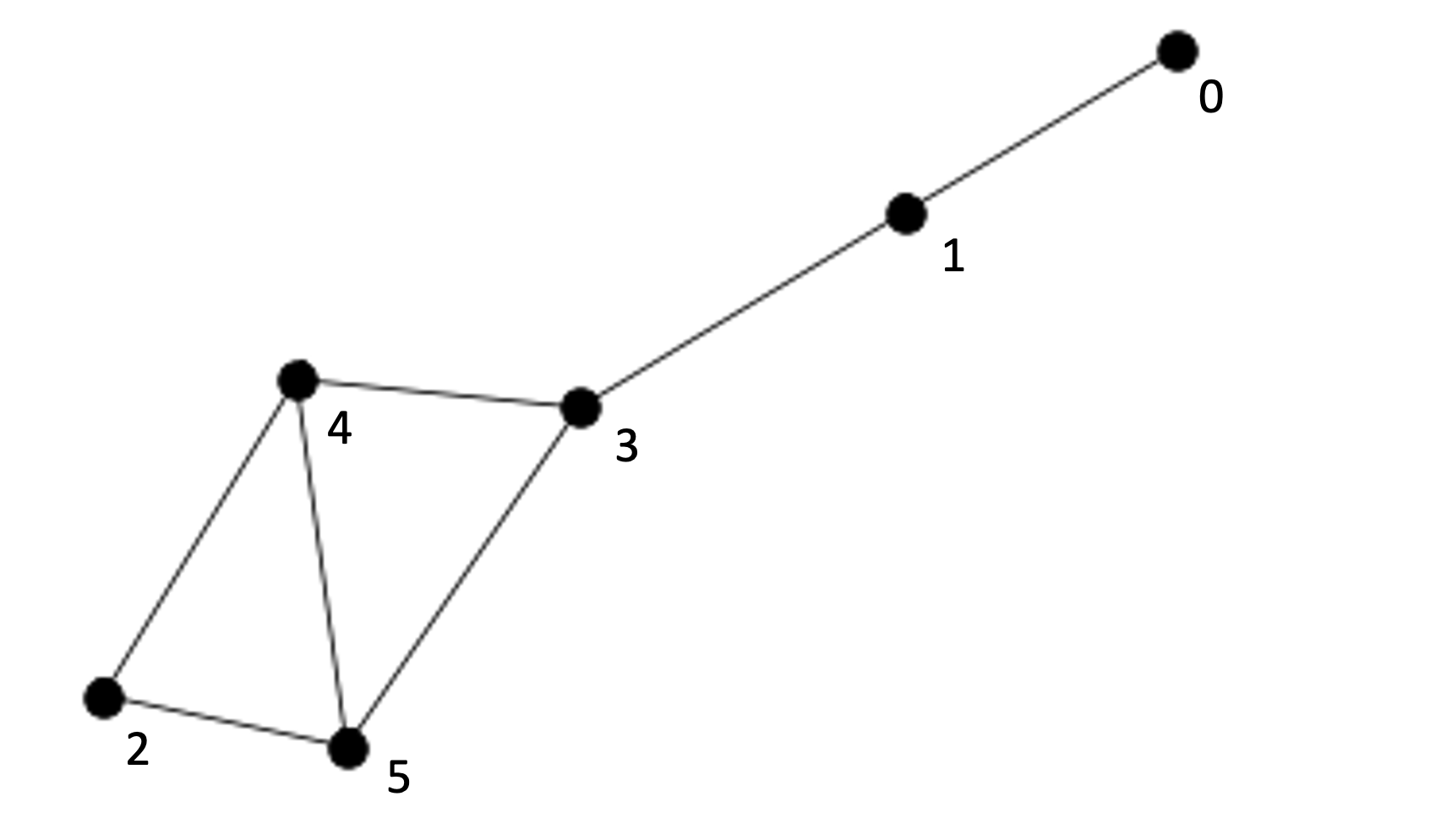}
    \caption{Test graph for dense subgraph code.}
    \label{fig:randomDSI}
\end{figure}

Taking $k=3$, when $\mathcal{S}=\{2\}$, we observed that, for a fixed number of runs, output samples corresponding  subgraph composed of vertices $2,4,5$ appeared the most number of times in the runs. Similarly, when $\mathcal{S}=\{4\}$, we observed that output samples of the induced subgraphs of vertices ${2,4,5}$ and ${3,4,5}$ appeared most, and with almost equal frequency. By direct inspection, it can be seen that our simulations did indeed manage to identify, for a given $\mathcal{S}$, the densest subgraph(s) of size $k$ which contains $\mathcal{S}$.

For graph isomorphism, our code estimates the Laplacian permanental polynomial of Eq.\ (\ref{eqlaplacian})  randomly chosen points $x$,  and for a user-chosen number  of samples.
As an application, we used this to successfully determine that the graphs  $GA$ and $GB$ shown in Figure\ \ref{fig:outputsPvL} are not isomorphic. The distinction is made by observing that for some value of $x$, the corresponding values of the Laplacian permanental polynomial of $GA$ and $GB$ did not match.
\begin{figure}[h]
     \centering
     \begin{subfigure}[b]{0.25\textwidth}
         \centering
         \includegraphics[width=0.9\textwidth]{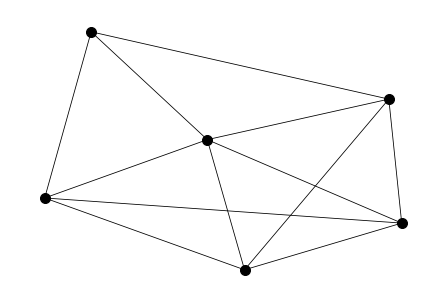} 
         \caption{}
     \end{subfigure}
     \begin{subfigure}[b]{0.25\textwidth}
         \centering
         \includegraphics[width=0.9\textwidth]{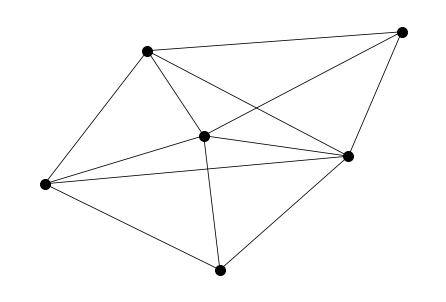}
         \caption{} 
     \end{subfigure}
\caption{Example of two non-isomorphic graphs of 6 vertices, $GA$ (Fig. 3(a)) and $GB$ (Fig. 3(b)). }
\label{fig:outputsPvL}
\end{figure}

As further application of our code and technique, we computed both the Laplacian permanental polynomial ($D_1)$  of Eq.(\ref{eqlappermpoly}) and the permanental polynomial ($D_2$) of Eq.\ (\ref{eqpermpolynomial}) and used these to distinguish non-isomorphic (or identify isomorphic) \emph{trees} \cite{merris1991almost}. We benchmarked the performance of these polynomials with an algorithm from  \cite{SciPyProceedings_11}  ($D_3$) which determines whether or not two graphs are isomorphic. 
We generated  $100$  pairs $(T^1_{i},T^2_{i})$ of random trees with $i \in \{1,..,100\}$  with $5$ vertices each, and used the distinguishers $D_1$, $D_2$, $D_3$ to classify, for each $i$, whether  $T^1_{i}$ is isomorphic (or not) to $T^2_{i}$. We obtained that for 31  pairs generated,  all three distinguishers outputted the same results, for 29 of the pairs only $D_1$ and  $D_3$ had same results, for 18 pairs only $D_2$ and $D_3$  had same results, and for 22 pairs neither $D_1$ nor $D_2$ outputted the same result as $D_3$. Our results agree with the fact that $D_1$ and $D_2$ are known to not be very good distinguishers of non-isomorphic trees \cite{merris1991almost}.

We also tested the performance of the distinguishers $D_1$ and $D_2$ for random graphs. We generated $100$ pairs of random graphs $(G^1_i,G^2_{i})$ of 5 vertices and edge probability $p=0.8$, with $i \in \{1,\dots,100\}$ and used $D_1$, $D_2$, $D_3$ to determine, for each $i$, whether (or not) $G^1_i$ is isomorphic to $G^2_{i}$. For $75$ pairs $D_1 - D_3$ outputted the same results, for 18 pairs only $D_1$ and $D_3$
outputted the same result, for 2 pairs only $D_2$
and $D_3$ outputted the same result, and for 5 pairs neither $D_1$ nor $D_2$ had the same result as $D_3$. This shows that our distinguishers are better at distinguishing random graphs than they are at distinguishing random trees. Finally, our performed tests show that  our distinguishers $D_1$ and $D_2$ have a comparable performance to the benchmark algorithm $D_3$. 

\section{Implementations on the Ascella Quantum Processor}
\label{sec:experiment}
We ran experiments on the cloud-accessible  \emph{Ascella} photonic quantum processor  \cite{ascella} \footnote{Quandela. Quandela cloud, 2022. \url{https://cloud.quandela.com}.}. The processor is composed of a fully-reconfigurable universal $12 \times 12$ linear optical circuit, a bright single-photon source coupled to a programmable optical demultiplexer producing up to 6 single photons, and single-photon detectors. Details about the optical setup as well as the single-photon source characteristics can be found in the supplementary material of \cite{ascella}.

The experiments performed consist of encoding graphs of $n$ vertices with $n \in \{3,4\}$ onto the linear optical circuit by the method of Section \ref{sec:encoding}. For each graph, we estimate the permanent of its adjacency  matrix $A$ using the output statistics of the device. The estimate is computed from $N=10000$ samples each corresponding to an 
event where $n$ photons are detected, of which $n_{post}$ are the post-selected samples corresponding to observing the events where $\mathbf{n_{in}}=\mathbf{n}$ (see Section \ref{sec:encoding}). Our estimate is then computed as
\begin{equation}
\label{eqestpermexp}
\mathsf{E}(\mathsf{Per}(A))=\sigma^n_{max}(A)\sqrt{\frac{n_{post}}{N}}.
\end{equation}

Our results are summarised in Table \ref{table:experiment}, where the exact value of the permanent of the adjacency matrix of each graph is also shown. Error bars are computed for a 95\% confidence interval using Hoeffding's inequality \cite{hoeffding1994probability}. The results show a good overlap between  estimated and exact values. Notably for some graphs tested , the interval $[\mathsf{E}(\mathsf{Per}(A))-\epsilon_{est},\mathsf{E}(\mathsf{Per}(A))+\epsilon_{est}] $ with $\epsilon_{est}$ the error bar does not contain the exact value of $\mathsf{Per}(A)$. This is likely a consequence of processor noise arising through single-photon distinguishability \cite{singphotrev}, multi-photon emissions \cite{singphotrev}, or imperfect  compilation \cite{errormitphot}. For a characterization of these errors for \emph{Ascella}, refer to the supplementary material of \cite{ascella}.

\begin{table}[h]
\centering
  \begin{tabular}{ | c | c | c | c |} \hline
      Graph & Exact Value & Estimated value & $n_{post}$\\ \hline\hline
     \includegraphics[scale=.1]{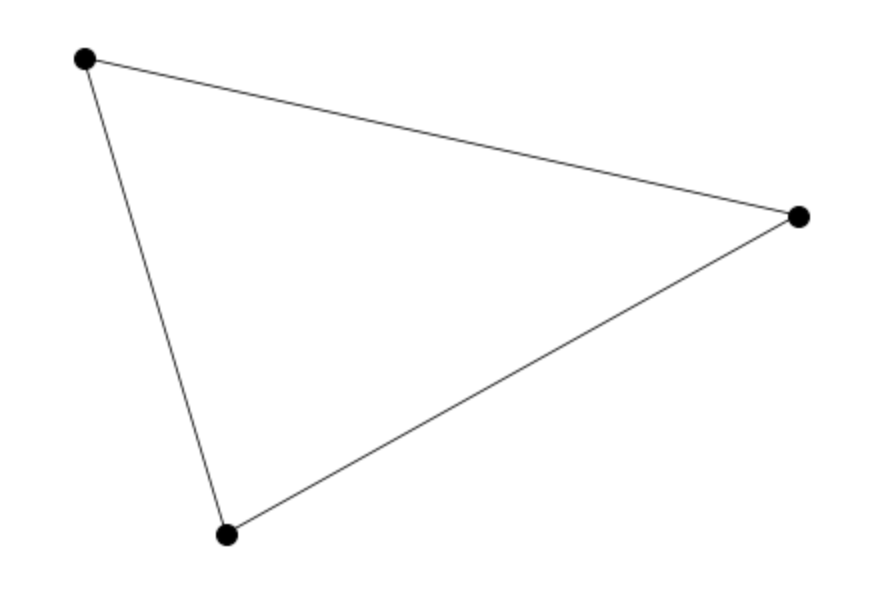} & $2$ & $1.841\pm 0.093$ & $1588$\\ \hline
     \includegraphics[scale=.1]{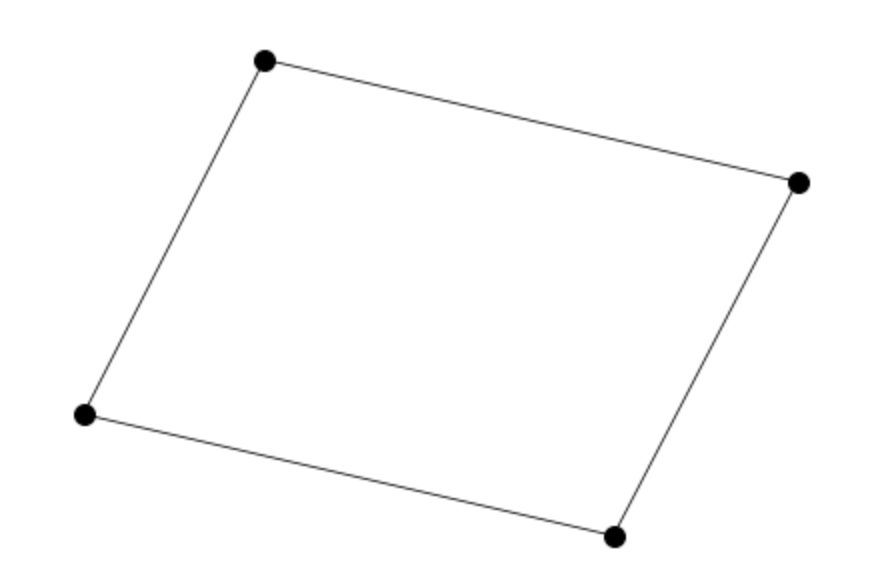} & $4$ & $3.512 \pm 0.186$ & $1927$\\ \hline
     \includegraphics[scale=.1]{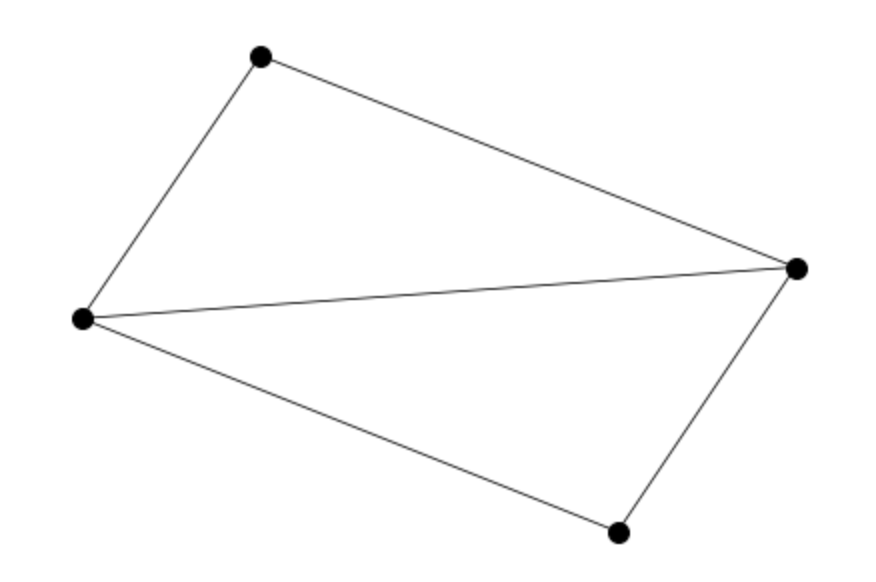} & $4$ & $3.967 \pm 0.500$ & $247$ \\ \hline
     \includegraphics[scale=.1]{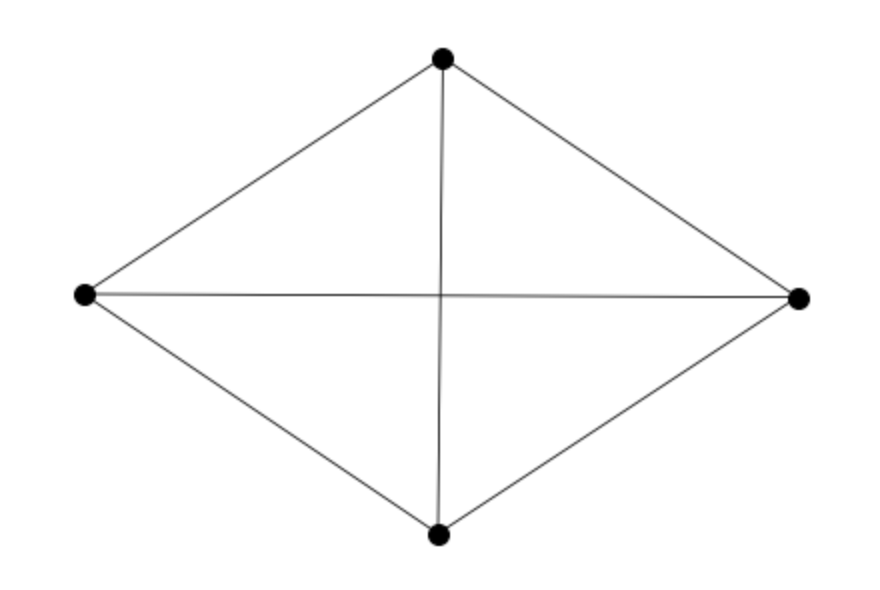} & $9$ & $8.512 \pm 0.941$ & $320$\\ \hline
  \end{tabular}
\caption{Results of  experiments performed on \emph{Ascella} for estimating the permanent of the adjacency matrix of several graphs of 3 and 4 vertices (see main text).}
\label{table:experiment}
\end{table}

\section{Discussion}
In summary, we have shown an efficient method for encoding a bounded $n \times n$ matrix onto a linear optical circuit of  $2n$ modes. We have shown how to use our encoding to solve various graph problems. We  performed numerical simulations validating our techniques. Finally, we performed experiments on photonic quantum hardware cementing the near-term utility of our developed techniques. 

Our work opens up possibilities for practical advantages \cite{coyle2021quantum,gonthier2022measurements}, in the sense that our methods outperform  specific classical strategies for some  instances of a given problem, and up to some (constant) input size. An interesting follow up question would be applying our methods to a specific use-case and highlighting the practical advantage obtained.

One might also ask whether our encoding could be used together with adaptive measurements \cite{chabaud2021quantum} to design new  photonic quantum algorithms  escaping the barrier of efficient classical simulability \cite{gurvits2005complexity}, and thereby presenting the potential for superpolynomial quantum speedups.

An interesting fact about our encoding is that it allows for computation of the permanent of any bounded matrix $A$, and not necessarily a symmetric matrix used for solving graph problems. As such, an interesting question would be identifying further problems whose solution can be linked to matrix permanents.

The unitaries used to encode matrices $A$ are not Haar-random, as can be seen from  Eq.(\ref{eqencUA}) for example. As such, one could hope that these unitaries could be implemented using linear optical quantum circuits of shallower depth than the standard universal interferometers \cite{Reck,Clements}. 
This is desirable in practice, as shallower circuits are naturally more robust to some errors such as photon loss \cite{oszmaniec2018classical,garcia2019simulating}. 

\begin{acknowledgements}
The authors thank Eric Bertasi, Alexia Salavrakos and Enguerrand Monard for contributions to the code; and
Andreas Fyrillas, Alexia Salavrakos, Luka Music, Arno Ricou, Jason Mueller, Pierre-Emmanuel Emeriau, Edouard Ivanov, and Jean Senellart for valuable discussions, comments and feedback.
We are grateful for support from the grant BPI France Concours Innovation PIA3 projects DOS0148634/00 and DOS0148633/00 – Reconfigurable Optical Quantum Computing.
\end{acknowledgements}

\bibliography{apssamp}

\appendix

\section{Notation}
\label{app:not}
We present here some  notation which we will use throughout this appendix.

We will denote by $G(V,E)$ (or sometimes $G$ for simplicity) a graph with a vertex set $V=\{v_1,...,v_n\}$ and edge set $E=\{e_1,...,e_l\}$, with $n,l \in \mathbb{N}^{*}$. The degree of vertex $v_i$ will be denoted as $|v_i|$, and is the number of edges connected to $v_i$. The adjacency matrix corresponding to $G$ will be denoted as $A(G)$ (or sometimes $A$ for simplicity). Unless otherwise 
specified, we will deal with unweighted, undirected, and simple graphs $G$. In these cases, the adjacency matrix $A(G)$ is a symmetric $(0,1)$-matrix \cite{graphbook}. The Laplacian  of a graph $G$  is defined as \cite{merris1981permanental}

\begin{equation}
    \label{eqlaplacian}
    L(G):=D(G)-A(G),
\end{equation}
with $D(G):=\mathsf{diag}(|v_1|,...,|v_n|)$ a diagonal matrix whose $ith$ entry is the degree of vertex $v_i$.

\bigskip
Let $\mathbf{E}_i$ be an $1 \times n$ row vector with zeros everywhere except at entry $i$ which is one. The $n \times n$ identity can then be written as
$$\mathbb{I}_{n \times n}=\begin{pmatrix}
\mathbf{E}_1\\
\mathbf{E}_2\\
.\\
.\\
.\\
\mathbf{E}_n
\end{pmatrix}.$$
Let $\pi:\{1,...,n\} \to \{1,...,n\}$ be a permutation of the set $\{1,...,n\}$, we will denote the \emph{symmetric group} or order $n$ (i.e, the set of all such permutations) as $\mathcal{S}_n$. The permutation matrix corresponding to $\pi$ is defined as
\begin{equation}
\label{eqpermmatrix}
    P_{\pi}:=\begin{pmatrix}
\mathbf{E}_{\pi(1)}\\
\mathbf{E}_{\pi(2)}\\
.\\
.\\
.\\
\mathbf{E}_{\pi(n)}
\end{pmatrix}.
\end{equation}
The set of all such permutation matrices forms a group, which we will denote as $\mathcal{P}_n$ \cite{weisstein2002permutation}.

\bigskip
Let $\mathcal{M}_n(\mathbb{C})$ be the set of $n \times n$ complex matrices, and let $M:=(M_{ij})_{i,j \in \{1,\dots,n\}} \in \mathcal{M}_n(\mathbb{C})$. We will denote by $\norm{M}$ the spectral norm of $M$,  defined as
\begin{equation}
    \label{eqspectralnorm}
    \norm{M}:=\sigma_{max}(M)=\sqrt{\lambda_{max}(M^{\dagger}M)},
\end{equation}
where $\sigma_{max}(M)$ is the largest singular value of $M$, which is equal to the square root of the largest eigenvalue of $M^{\dagger}M$, denoted as $\lambda_{max}(M^{\dagger}M)$; $M^{\dagger}$ denotes the conjugate transpose of $M$. $M^T$ will denote the transpose of $M$. Also, let
\begin{equation}
  \norm{M}_{\infty}:=\mathsf{max}_{i}\sum_{j=1,\dots,n}|M_{ij}|,
\end{equation}
where $\mathsf{max}_i$ denotes the maximum of the above defined sum over all rows $i \in \{1,\dots,n\}$ of $M$. As well as
\begin{equation}
     \norm{M}_{1}:=\mathsf{max}_{j}\sum_{i=1,\dots,n}|M_{ij}|,
\end{equation}
where $\mathsf{max}_j$ denotes the maximum of the above defined sum over all columns $j \in \{1,\dots,n\}$ of $M$.

\section{Detailed comparision with previous work}
\label{app:comp}
The main differences between our encoding and that of \cite{bradler1,bradler2,arrazola,schuld}, which in general also encodes a (real symmetric)  $n \times n$ matrix $A$ into a photonic setup with $2n$ modes, are $(1)$ our encoding directly embeds $A$  into a linear optical circuit, whereas the encoding in \cite{bradler1,bradler2,arrazola,schuld} encodes $A$ by using a combination of squeezed  states of light, as well as linear optical circuits; and $(2)$ our encoding works also for general non-symmetric bounded matrices $A$, whereas that of \cite{bradler1,bradler2,arrazola,schuld} supports only symmetric matrices $A$. Of course, there are ways to construct, starting from non-symmetric $A$, a larger  matrix $\mathcal{A}$  which is symmetric \footnote{$\mathcal{A}=\begin{pmatrix}
0_{n \times n} & A \\ A^T & 0_{n \times n}
\end{pmatrix}$, is an example of such a construction.}, then encoding $\mathcal{A}$ using techniques in \cite{bradler1,bradler2,arrazola,schuld}. However,  this requires using a photonic setup of $L>2n$ modes, and it is unclear whether the number of modes could be reduced back to $2n$ in this setting. Finally, $(3)$ our photonic setup composed of single-photon sources, linear optical circuits, and single-photon detectors, when used together with our encoding naturally allows the computation of the permanent of a matrix, whereas the setup in \cite{bradler1,bradler2,arrazola,schuld} computes the Hafnian ($\mathsf{Haf}(.)$) of a matrix \cite{rudelson2016hafnians, caianiello1953quantum}. Although the Hafnian is in some sense a generalization of the permanent, since
\begin{equation}
    \label{eqhafnian}
    \mathsf{Haf}\begin{pmatrix} 0_{n \times n} & A \\ A^T & 0_{n \times n}  \end{pmatrix}=\mathsf{Per}(A),
\end{equation}
where  $0_{n \times n}$ is the all-zeros $n \times n$ matrix.
 Nevertheless, Eq.(\ref{eqhafnian}) highlights the fact that, using the setup in \cite{bradler1,bradler2,arrazola,schuld} together with their encoding to compute $\mathsf{Per}(A)$, requires in general a number of modes exactly double that needed to compute $\mathsf{Per}(A)$ using our setup. To explain this point further, first note that, although $\mathcal{A}=\begin{pmatrix} 0_{n \times n} & A \\ A^T & 0_{n \times n}  \end{pmatrix}$ is symmetric, it does not satisfy the criteria of encodability onto a Gaussian state mentioned in \cite{bradler1}, since the off-diagonal blocks need to be equal as well as positive definite. Thus, one needs to use $\mathcal{A}\bigoplus\mathcal{A}$ which maps onto a Gaussian covariance matrix \cite{bradler1}, but this is a $4n \times 4n$  matrix. Finally, note that input states other than squeezed states,  such as thermal states, have been used in Gaussian Boson Sampling for encoding and computing the permanent of positive definite matrices \cite{jahangiri2020point}.
 \section{Computing permanental polynomials}
\label{app:permpoly}
In order  to compute the coefficients $\{c_i\}$ in Eq.(\ref{eqpermpolynomial}), we perform $n+1$ experiments, where in each experiment we encode $B_{x}$ into a linear optical circuit, and then estimate $\mathsf{Per}(B_{x})$. For each experiment $i$,
 we choose a different value $x_i$ of $x$, for $i$ going from 1 to $n+1$. By doing this, we obtain the following system of $n+1$ linear equations in $n+1$ unknowns $c_0,\dots, c_n$ 
 \begin{equation}
     \label{eqlinearsystem}
     \begin{pmatrix}
      1 & x_1 & \dots & x^n_{1} \\
      1 & x_2  & \dots & x^n_{2}
      \\ & \\       &    \\
      .  & .  & \dots &  . \\
      \\    \\
      .  & .   & \dots&. \\
      \\    \\
      .\\ &  .    & \dots&. \\
      1 & x_{n+1} & \dots & x^n_{n+1}
     \end{pmatrix} \begin{pmatrix} c_0 \\ c_1  \\   . \\.  \\.  \\  c_n \end{pmatrix}= \begin{pmatrix} P_A(x_1) \\ P_A(x_2) \\. \\. \\. \\ P_A(x_{n+1}) \end{pmatrix}.
 \end{equation}
 Let 
 \begin{equation*}
 D(x_1,...,x_{n+1}):=\begin{pmatrix}
      1 & x_1 & \dots & x^n_{1} \\
      1 & x_2  & \dots & x^n_{2}
      \\ & \\       &    \\
      .  & .  & \dots &  . \\
      \\    \\
      .  & .   & \dots&. \\
      \\    \\
      .\\ &  .    & \dots&. \\
      1 & x_{n+1} & \dots & x^n_{n+1} \end{pmatrix}.
 \end{equation*}
 The determinant 
 \begin{equation}
     \label{eqdet}
     f(x_1,...,x_{n+1}):=\mathsf{Det}(D(x_1,...,x_{n+1})),
 \end{equation}
 is a polynomial of $(x_1,...,x_{n+1}) \in \mathbb{R}^{n+1}$ which is non-identically zero, thus we can make use of the following lemma proven in \cite{okamoto1973distinctness}.
 \begin{lemma} 
 \label{lemnonidzero}
 Let $f(x_1,...,x_{n+1})$ be a polynomial of real variables $(x_1,...,x_{n+1}) \in \mathbb{R}^{n+1}$ which is non-identically zero. Then the set $\{(x_1,...,x_{n+1}) \mid f(x_1,...,x_{n+1})=0\}$ has Lebesgue measure zero in $\mathbb{R}^{n+1} $.
 \end{lemma}
 Lemma \ref{lemnonidzero} implies that \emph{almost any} choice of $(x_1,...,x_{n+1})$ gives an invertible matrix $D(x_1,...,x_{n+1})$, since its determinant is non-zero for almost any choice (except a set of measure zero) of $(x_1,...,x_{n+1})$. This is important, as it allows one to solve the system of linear equations in Eq.(\ref{eqlinearsystem}) with high probability by randomly choosing $n+1$ values of $x$, and thereby determine the coefficients $\{c_i\}$ of the permanental polynomial.

 As a final remark, note that our setup allows estimating $|\mathsf{Per}(B_x)|$, rather than $\mathsf{Per}(B_x)$ needed to solve the system of linear equations. We can however, knowing the sign of $\mathsf{Per}(B_x)$, always deduce it from $|\mathsf{Per}(B_x)|$. Choosing $x \in \mathbb{R}^{-}$, gives
 $\mathsf{Per}(B_x)=(-1)^n\mathsf{Per}(-x\mathbb{I}_{n \times n}+A)$, where $\mathsf{Per}(-x\mathbb{I}_{n \times n}+A) \geq 0$. In this way we can always know the sign of $\mathsf{Per}(B_x)$ beforehand. By lemma \ref{lemnonidzero}, choosing points of the form $(x_1,\dots,x_{n+1})$ with $x_i \leq 0$ allows for solving the system of linear equations, since the set of these points does not have measure zero in $\mathbb{R}^{n+1}$.

\section{$k$-densest subgraph problem}
\label{app:kdensest}
In this section we prove Theorem \ref{thperm} which we restate here for convenience.

Let $G(V,E)$ be a graph with $|V|=n$, $|E|=I$, with $n,I \in \mathbb{N}^{*}$, and $n, I$ even. Let $A=(a_{ij})_{i,j \in \{1,\dots,n\}}$, with $a_{ij} \in \{0,1\}$ be the adjacency matrix of $G$. Theorem \ref{thperm} states that
\begin{equation*}
\mathsf{Per}(A) \leq f(n,I),
\end{equation*}
where $f(n,I)$ is a monotonically increasing function with increasing $I$, for fixed $n$.

\begin{proof}
Let $r_i=\sum_{j=1 \dots n}a_{ij}.$ Consider the upper bound for $\mathsf{Per}(A)$ for a $(0,1)$-matrix $A$ shown in \cite{Minc,Bergman}
\begin{equation}
\label{eqlemm21}
    \mathsf{Per}(A) \leq \prod_{i=1,\dots n}(r_i !)^{\frac{1}{r_i}}.
\end{equation}
Also, note the following upper bound shown in \cite{Agha} for simple graphs $G$ with even $n$ and $I$
\begin{equation}
    \label{eqlemm22}
    \prod_{i=1,\dots n}(r_i !)^{\frac{1}{2r_i}} \leq \omega(n, I),
\end{equation}
with 
\begin{equation}
    \omega(n,I):= \left( \biggl\lfloor \frac{2 I}{n}  \biggr\rfloor !\right)^{\frac{\frac{n}{2} - \alpha}{\bigl\lfloor \frac{2 I}{n}\bigr\rfloor} } \left( \biggl\lceil \frac{2 I}{n}  \biggr\rceil  ! \right)^{\frac{\alpha}{\bigl\lceil \frac{2 I}{n}\bigr\rceil} },
\end{equation}
with
\begin{equation}
    \alpha:=I-n\biggl\lfloor \frac{I}{n} \biggr\rfloor,
\end{equation}
and $\lceil . \rceil$, $\lfloor . \rfloor$ denoting the ceiling and floor functions respectively.
Taking the square root of Eq.(\ref{eqlemm21}) and plugging it in Eq.(\ref{eqlemm22}), we get
\begin{equation}
\label{eqlemm23}
  \sqrt{\mathsf{Per}(A)} \leq \omega(n, I).  
\end{equation}
Squaring Eq.(\ref{eqlemm23}), then defining $f(n, I):=(\omega(n, I))^2$, while noting that $\omega(n, I)$ (and therefore $f(n, I)$) is monotonically increasing with  increasing $I$ for fixed $n$, as observed in \cite{arrazola}, completes the proof.
\end{proof}

\section{Graph isomorphism}
\label{app:GI}

Let $A$ and $B$ be the adjacency matrices of two (unweighted, undirected, no self loops) graphs $G_1$ and $G_2$ with $n$ vertices each. We will also assume that $G_1$ and $G_2$ are \emph{isospectral}, that is they have the same eigenvalues. Isomorphic graphs are also isospectral, this can be seen by noting that, if $B=P_{\pi}AP^{T}_{\pi}$, then the characteristic polynomials of $A$ and $B$ are equal. That is,
\begin{multline*}
\mathsf{Det}(\lambda \mathbb{I}_{n \times n}-B)=\mathsf{Det}(\lambda \mathbb{I}_{n \times n}-P_{\pi}AP^{T}_{\pi})=\\ \mathsf{Det}(P_{\pi}(\lambda \mathbb{I}_{n \times n}-A)P^T_{\pi})=\\\mathsf{Det}(P_{\pi}P^T_{\pi})\mathsf{Det}(\lambda \mathbb{I}_{n \times n}-A))=\mathsf{Det}(\lambda \mathbb{I}_{n \times n}-A)),
\end{multline*}
since $P_{\pi}P^T_{\pi}=\mathbb{I}_{n \times n}$. The converse however, that isospectral graphs are isomorphic, is not true \cite{beineke1981spectra}. Since determining the eigenvalues of an $n \times n$ matrix takes $O(n^3)$-time \cite{pan1999complexity}, it is good practice to check whether $G_1$ and $G_2$ are isospectral before proceeding to check if they are isomorphic, as there is no point in continuing if they are not isospectral.

We will now prove Theorem \ref{th1} in the main text.

\begin{proof}

\textbf{Proof that 2) $\implies$ 1)}

$G_1$ is isomorphic to $G_2$, then $B=P_{\pi}AP^T_{\pi}$, with $P_{\pi} \in \mathcal{P}_n$. Writing $A$ as
$A=(a_{ij})_{i,j \in \{1,...,n\}}$, we can write $B$ as $B=(b_{ij})_{i,j \in \{1,\dots,n\} }=(a_{\pi(i)\pi(j)})_{i,j \in \{1,\dots,n\}},$ with $\pi: \{1,\dots,n\} \to \{1,\dots,n\}$ the bijection corresponding to $P_{\pi}$. That $B$ can be written this way can be seen directly by noting that $P_{\pi}$ (respectively $P^T_{\pi}$) permutes the rows (respectively columns) of $A$ according to $\pi$.
For $l \in \{1,\dots,n\}$, $\mathbf{s}=\{s_1,\dots,s_l\}$, $\mathbf{t}=\{t_1,\dots,t_l\}$, the submatrix $B_{\mathbf{t},\mathbf{s}}$ is given by
\begin{multline*}
   B_{\mathbf{t},\mathbf{s}}=\begin{pmatrix} b_{s_{1}t_1} & \dots & b_{s_{1}t_l} \\
   .
   \\.
   \\.
   \\
   b_{s_{l}t_1} & \dots & b_{s_{l}t_l}
   \end{pmatrix} = \\ \begin{pmatrix} a_{\pi(s_{1})\pi(t_1)} & \dots & a_{\pi(s_{1})\pi(t_l)} \\
   .
   \\.
   \\.
   \\
   a_{\pi(s_{l})\pi(t_1)} & \dots & a_{\pi(s_{l})\pi(t_l)}
   \end{pmatrix} =A_{\mathbf{\pi(t)},\mathbf{\pi(s)}}.
\end{multline*}
Thus, $\mathsf{Per}(A_{\mathbf{\pi(t)},\mathbf{\pi(s)}})=\mathsf{Per}(B_{\mathbf{t},\mathbf{s}})$, and this holds $\forall$ $l,\mathbf{s}, \mathbf{t}$. Therefore, we recover statement 1).

\textbf{Proof that 1) $\implies$ 2)}
We have that $\forall$ $l,\mathbf{s},\mathbf{t}$, $\mathsf{Per}(B_{s,t})=\mathsf{Per}(A_{\mathbf{\pi(t)},\mathbf{\pi(s)}})$. In particular, consider the case where $\mathbf{s}=\{i,\dots,i\}$, $\mathbf{t}=\{j,\dots,j\}$, with $i,j \in \{1,\dots,n\}$.
We then have
\begin{multline*}
    \mathsf{Per}(B_{s,t})=b^l_{ij}\mathsf{Per} \begin{pmatrix} 1 & 1 \dots & 1
    \\
    1 & 1 \dots & 1
    \\&.
   \\&.
   \\ &. \\
   1 & 1 \dots & 1 \end{pmatrix}=\\\mathsf{Per}(A_{\mathbf{\pi(t)},\mathbf{\pi(s)}})=a^l_{\pi(i)\pi(j)}\mathsf{Per} \begin{pmatrix} 1 & 1 \dots & 1
    \\
    1 & 1 \dots & 1
    \\&.
   \\&.
   \\ &. \\
   1 & 1 \dots & 1 \end{pmatrix}.
\end{multline*}
Thus
\begin{equation*}
    b^l_{ij}=a^l_{\pi(i)\pi(j)},
\end{equation*}
which holds $\forall$ $l$, where $\pi$ is a fixed bijection.
Since $G_1$, $G_2$ are unweighted and undirected, this means that $a_{\pi(i)\pi(j)},b_{ij} \in \{0,1\}$, and therefore that
\begin{equation*}
    b_{ij}=a_{{\pi(i)\pi(j)}},
\end{equation*}
which holds $\forall$ $l,i,j \in \{1,\dots,n\}$, and where $\pi$ is fixed.

Therefore, we can deduce that $B=(a_{\pi(i)\pi(j)})_{i,j \in \{1,\dots,n\}}=P_{\pi}AP^T_{\pi}$. We have thus recovered statement 2).

This completes the proof of Theorem \ref{th1}.

\end{proof}

As already mentioned in the main text, and made concrete through Theorem \ref{th1},  we have shown that our setup provides necessary and sufficient conditions for two graphs to be isomorphic. However, the number of experiments we need to perform  scales exponentially with the number of vertices of the graphs (see main text).   To get around this, we can instead choose to compute Laplacian permanental polynomials (Eq.(\ref{eqlappermpoly})), which are powerful distinguishers on non-isomorphic graphs \cite{merris1981permanental}.  
We now prove the following lemma, which is probably found in the literature, showing that isomorphic graphs have the same Laplacian permanental polynomials.
\begin{lemma}
\label{eqlemlapiso}
Let $G_1$ and $G_2$ be two isomorphic graphs with adjacency matrices $A$, $B$, where $B=P_{\pi}AP^T_{\pi}$, with $P_{\pi} \in \mathcal{P}_n$. Let $L(G_1)$ and $L(G_2)$ be the Laplacians of $G_1$ and $G_2$, then $L(G_2)=P_{\pi}L(G_1)P^T_{\pi}$, and furthermore $\mathsf{Per}(x \mathbb{I}_{n \times n}-L(G_1))=\mathsf{Per}(x \mathbb{I}_{n \times n}-L(G_2))$, for all $x \in \mathbb{R}$.
\end{lemma}
\begin{proof}
$L(G_2)=D(G_2)-B$, with $B=P_{\pi}AP^T_{\pi}$, and $D(G_2)=(d(G_2)_{ii})_{i \in \{1,\dots,n\}}$, with $d(G_2)_{ii}$ degree of vertex $i$ of $G_2$, which is vertex $\pi(i)$ of $G_1$. Thus $D(G_2)=(d(G_1)_{\pi(i)\pi(i)})_{i \in \{1,\dots,n\}}=P_{\pi}D(G_1)P^T_{\pi}$, and consequently,  $L(G_2)=P_{\pi}L(G_1)P^T_{\pi}$. Using this, we have that 
\begin{multline*}
   \mathsf{Per}(x \mathbb{I}_{n \times n} -L(G_2))=\mathsf{Per}(x \mathbb{I}_{n \times n} -P_{\pi}L(G_1)P^T_{\pi})=\\ \mathsf{Per}(P_{\pi}(x \mathbb{I}_{n \times n} -L(G_1))P^T_{\pi})=\\ \mathsf{Per}(x \mathbb{I}_{n \times n} -L(G_1)),
\end{multline*}
where the last equality holds from the fact that the permanent is invariant under permutations \cite{botta1967linear}. This concludes the proof.
\end{proof}

Computing the coefficients of Laplacian permanental polynomials can be done using our setup, in a similar way to how these coefficients are computed for permanental polynomials, as seen in Section \ref{app:permpoly}. Indeed, replacing $B_x=x \mathbb{I}_{n \times n}-A$  in Section \ref{app:permpoly}, with $B_x=x \mathbb{I}_{n \times n}-L(G)$, then following the same steps as in Section  \ref{app:permpoly} allows one to compute the coefficients of the Laplacian permanental polynomial.

\section{Boosting output probabilities}
\label{app:boosting}
\subsection*{First method for boosting}
Consider the matrix $$A=\begin{pmatrix}
    \mathbf{A}_1 \\ .\\.\\.\\ \\ \mathbf{A}_n
    \end{pmatrix},$$ with $\mathbf{A}_i=(a_{i1},...,a_{in})$ the $ith$ row vector of $A \in \mathcal{M}_n(\mathbb{R})$. We will first discuss the method where we attempt to boost the probability of appearance of the output corresponding to $\mathsf{Per}(A)$ in our setup  by modifying $A$ as follows.
    Let 
    \begin{equation}
        A_w=\begin{pmatrix}
    \mathbf{A}_1 \\ .\\.\\.\\ \\ \mathbf{A}_{c-1} \\ w\mathbf{A}_c\\ \mathbf{A}_{c+1} \\ .\\ . \\.\\  \\ \mathbf{A}_n
    \end{pmatrix},
    \end{equation}
    where the $c$th row of $A$ is multiplied by $w \in \mathbb{R}^{+*}$. 
    We first prove the following lemma.
    \begin{lemma}
       \label{lemboosting1}
       $\mathsf{Per}(A_w)=w\mathsf{Per}(A)$.
    \end{lemma}
    \begin{proof}
       Let $A=(a_{ij})_{i,j \in \{1,\dots,n\}}$, $A_w=(b_{ij})_{i,j \in \{1,\dots,n\}}$. Looking at Eq.(\ref{eqpermexpansion}) for $\mathsf{Per}(A_w)$, an element the $c$th row appears exactly once in each product $\prod_{i=1,\dots,n}b_{ i\pi(i)}$ in the sum. Since $b_{c\pi(c)}=w a_{c \pi(c)}$ Thus, $\prod_{i=1,\dots,n}b_{i\pi(i)}=w\prod_{i=1,\dots,n}a_{i\pi(i)}$. Thus,
       $\sum_{\pi \in \mathcal{S}_n}\prod_{i=1,\dots,n}b_{i\pi(i)}=w\sum_{\pi \in \mathcal{S}_n}\prod_{i=1,\dots,n}a_{i\pi(i)}$, which completes the proof.
    \end{proof}
   Lemma \ref{lemboosting1} allows one to efficiently compute an estimate of $\mathsf{Per}(A)$, given an estimate of $\mathsf{Per}(A_w)$.

   Let $p(\mathbf{n}|\mathbf{n_{in}})$ (respectively $p_w(\mathbf{n}|\mathbf{n_{in}})$) be the probabilities of observing outcomes corresponding to $\mathsf{Per}(A)$ (respectively $\mathsf{Per}(A_w)$) in our setup. 
  Boosting happens when 
  \begin{equation}
  \label{eqboosting2}
   p_w(\mathbf{n}|\mathbf{n_{in}})>  p(\mathbf{n}|\mathbf{n_{in}}).
   \end{equation}
   We will now show the following
   \begin{lemma}
      \label{lem2boosting}
     $ p_w(\mathbf{n}|\mathbf{n_{in}})>  p(\mathbf{n}|\mathbf{n_{in}})$ $\implies$
      \begin{equation}
      \label{eqboosting3}
      \frac{\sigma_{max}(A)}{\sigma_{max}(A_w)} > \frac{1}{w^{\frac{1}{n}}}.
      \end{equation}
   \end{lemma}
   \begin{proof}
   Plugging Eq.(\ref{eqestimateperm}) in Eq.(\ref{eqboosting2}), while choosing $s=\sigma_{max}(A)$ (respectively $\sigma_{max}(A_w)$) for $p(\mathbf{n}|\mathbf{n_{in}})$ (respectively $p_w(\mathbf{n}|\mathbf{n_{in}})$)
   , and using Lemma \ref{lemboosting1} gives 
  \begin{equation*}
   w^2\frac{|\mathsf{Per}(A)|^2}{\sigma^{2n}_{max}(A_w)}  > \frac{|\mathsf{Per}(A)|^2}{\sigma^{2n}_{max}(A)}.
  \end{equation*}
   Assuming $\mathsf{Per}(A) \neq 0$, allows for removing it from both sides of the above equation. Regrouping the terms in the above equation, and taking it to the $(.)^{\frac{1}{2n}}$ power completes the proof.
   \end{proof}
   Although lemma \ref{lem2boosting} gives a necessary condition for boosting to occur, it is not very informative as it does not answer the question: what properties should $A$  verify for boosting to be possible under our above defined modification ? It will be the aim of the rest of this section to dig deeper in an attempt to answer the above question.
   
   Let
   
   \begin{equation}
       \label{eqboostinggamma}
       \gamma:=\sum_{j}|a_{cj}|,
   \end{equation}
   with $a_{cj}$ the element of the $c$th row and $jth$ column of $A$.
   Note that
   \begin{equation*}
       \gamma \leq \norm{A}_{\infty},
   \end{equation*}
   and also that
   \begin{equation}
   \label{eqrelsigmaainfty}
   \sigma_{max}(A) \leq \sqrt{n}\norm{A}_{\infty},
   \end{equation}
   where this last equation follows immediately from the well known relation \cite{golub2013matrix}
   \begin{equation}
   \label{eqidtyainftya}
      \norm{A} \leq \sqrt{n} \norm{A}_{\infty}.
   \end{equation}
   From the definition of $\gamma$, $\norm{.}_{\infty}$, and $A_w$ , we also have
   \begin{equation}
       \label{eqinftaw}
       \norm{A_w}_{\infty}=\mathsf{max}(w\gamma,\norm{A}_{\infty}).
   \end{equation}
   Finally, recall the following relation for the trace (denoted as $\mathsf{Trace}(.)$) of matrices $L=(l_{ij})_{i,j \in \{1,\dots,n\}}$ and $M=(m_{ij})_{i,j \in \{1,\dots,n\}}$
   \begin{equation}
       \label{eqtrab}
       \mathsf{Trace}(LM)=\sum_{i=1,\dots,n}\sum_{j=1,\dots,n}l_{ij}m_{ji}.
   \end{equation}
   
   With the above equations in hand we will now prove the following.
   \begin{lemma}
      \label{lemboundsA}
      $ \sqrt{\frac{\sum_{i}\sigma^2_i(A)}{n}}\leq \sigma_{max}(A) \leq \sqrt{n}\norm{A}_{\infty}$, with $\{\sigma_{i}(A)\}_{i \in \{1,\dots,n\}}$ the singular values \footnote{Note that $\sigma_{max}(A)=\sigma_k(A)$ for some $k \in \{1,\dots,n\}$.} of $A$ .
   \end{lemma}
   \begin{proof}
   The upper bound on $\sigma_{max}(A)$ follows immediately from Eq.(\ref{eqrelsigmaainfty}). For the lower bound, we begin by noting, from the definition of singular values of $A$, that
   \begin{equation*}
       \mathsf{Trace}(AA^T)=\sum_{i=1,\dots,n}\sigma^2_{i}(A).
   \end{equation*}
   Using the fact that  $\sigma_{i}(A) \leq \sigma_{max}(A)$, $\forall$ $i \in \{1,\dots,n\}$, and plugging this into the above equation gives
   \begin{equation}
   \label{eqfininproof}
     \sum_{i=1,\dots,n}\sigma^2_{i}(A)  \leq n \sigma^2_{max}(A).
   \end{equation}
   Rearranging the terms in Eq.(\ref{eqfininproof}), and taking the square root of it gives the desired lower bound and completes the proof.
   \end{proof}
   For $A_w$, we show that
   \begin{lemma}
      \label{lemboundsAw}
     $\sqrt{\frac{\sum_{i}\sigma^2_{i}(A)+(w^2-1)\delta}{n}} \leq  \sigma_{max}(A_w) \leq \sqrt{n}\mathsf{max}(w\gamma,\norm{A}_{\infty})$, with $\delta:=\sum_{j}a^2_{cj}.$
       \end{lemma}
       \begin{proof}
       The upper bound for $\sigma_{max}(A_w)$ follows from plugging Eq.(\ref{eqinftaw})  in the relation $\sigma_{max}(A_w) \leq \sqrt{n}\norm{A_w}_{\infty}$. For the lower bound, denote $A=(a_{ij})$,  $A_w=(b_{ij})$, and consider
       \begin{equation*}
         \mathsf{Trace}(A_wA^T_w)=\sum_{i}\sum_{j}b_{ij}^{2}=\sum_{i \neq c}\sum_{j}a^2_{ij}+w^2\sum_{j}a^2_{cj},
       \end{equation*}
       where the second equality follows from using the relation of Eq.(\ref{eqtrab}), and the third equality follows from noting that $b_{ij}=a_{ij}$ for $i \neq c$, and $b_{cj}=wa_{cj}$. Now,

       \begin{align*}
       \sum_{i \neq c}\sum_{j}a^2_{ij}+w^2\sum_{j}a^2_{cj}&=
       \sum_{i}\sum_{j}a^2_{ij}+(w^2-1)\sum_{j}a^2_{cj}\\
       &=\mathsf{Trace}(AA^T)+(w^2-1)\delta\\
       &=\sum_{i}\sigma^2_i(A)+(w^2-1)\delta.
       \end{align*}

       Thus
       \begin{equation}
           \label{eqexpaw}
           \mathsf{Trace}(A_wA^T_w)=\sum_{i}\sigma^2_i(A)+(w^2-1)\delta.
           \end{equation}
           Noting that
           \begin{equation*}
              \mathsf{Trace}(A_wA^T_w) \leq n \sigma^2_{max}(A_w),
           \end{equation*}
           then plugging this into Eq.(\ref{eqexpaw}), rearranging, and taking the square root, one obtains the desired lower bound for $\sigma_{max}(A_w)$. This completes the proof.
           \end{proof}
           
          Taking $w>1$, and with lemmas \ref{lemboundsA} and \ref{lemboundsAw} in hand, we can make the following observations. First, if
          
          \begin{equation}
             \label{eqcondi1}
             w\gamma < \norm{A}_{\infty},
          \end{equation}
          the upper bounds of $\sigma_{max}(A)$ and $\sigma_{max}(A_w)$ coincide. Furthermore, if
          \begin{equation}
              \label{eqcondi2}
              (w^2-1)\delta \ll \sum_{i}\sigma^2_{i}(A)=\mathsf{Trace}(AA^T),
          \end{equation}
          then the lower bounds of $\sigma_{max}(A)$ and $\sigma_{max}(A_w)$ almost coincide.
          
          Verifying the conditions in equations (\ref{eqcondi1}) and  (\ref{eqcondi2}) for some values of $w$ and $n$ likely implies that $\sigma_{max}(A_w) \approx \sigma_{max}(A)$, and therefore that the condition

         \begin{equation*}
    \frac{\sigma_{max}(A)}{\sigma_{max}(A_w)} > \frac{1}{w^{\frac{1}{n}}},
          \end{equation*}
          is satisfied, which, from lemma \ref{lem2boosting}, is a necessary condition for boosting. Since $\delta$, $\gamma$, $\norm{A}_{\infty}$, and $\mathsf{Trace}(AA^T)$, are properties of $A$ which are easily computable. We have thus established a way to test whether boosting using our technique is possible, given some matrix $A$.
          
          What remains is to find matrices $A$ satisfying the above properties (equations (\ref{eqcondi1}) and (\ref{eqcondi2})) for some $w$, and some values of $n$. One example which we, numerically, find satisfies these properties, and for which we observe boosting is the adjacency matrix of the ten vertex graph 
          \begin{equation}
              \label{eqadjmatrixexboosting}
              A=\begin{pmatrix} 0& 1 & 1 & 1 & 1 & 1 & 1 & 1 & 1 & 0 \\ 1& 0 & 1 & 1 & 1 & 1 & 1 & 1 & 1 & 1 \\ 1& 1 & 0 & 1 & 1 & 1 & 1 & 1 & 1 & 1 \\1& 1 & 1 & 0 & 1 & 1 & 1 & 1 & 1 & 0 \\ 1& 1 & 1 & 1 & 0 & 1 & 1 & 1 & 1 & 0 \\ 1& 1 & 1 & 1 & 1 & 0 & 1 & 1 & 1 & 0 \\ 1& 1 & 1 & 1 & 1 & 1 & 0 & 1 & 1 & 0 \\ 1& 1 & 1 & 1 & 1 & 1 & 1 & 0 & 1 & 0 \\ 1& 1 & 1 & 1 & 1 & 1 & 1 & 1 & 0 & 0 \\ 0&1 & 1 & 0 & 0 & 0 & 0 & 0 & 0 & 0
              \end{pmatrix},
          \end{equation}
          and where we choose $c=10$ when constructing $A_w$ (i.e we multiply the tenth row of $A$ by $w$). 
          The graph corresponding to $A$ is represented in Figure \ref{fig:isolated}. The conditions in equations (\ref{eqcondi1}) and (\ref{eqcondi2}) appear to be satisfied, up to a certain value of $w$, whenever row  $c$ corresponds to a vertex which has a significantly lesser degree than other vertices in the graph, as can be seen in the above chosen example.
          
          \begin{figure}[h]
              \centering
              \includegraphics[scale=0.3]{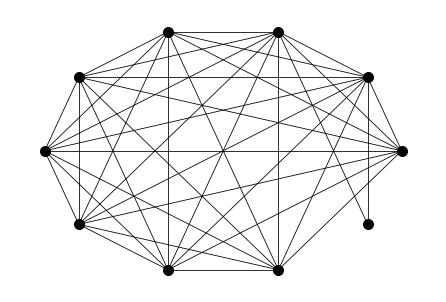}
              \caption{The graph with adjacency matrix $A$ in Eq.(\ref{eqadjmatrixexboosting}).}
              \label{fig:isolated}
          \end{figure}

          Let 
          
          \begin{equation}
              \label{eqratioboost}
        \mathcal{R}:=\frac{p_w(\mathbf{n}|\mathbf{n_{in}})}{p(\mathbf{n}|\mathbf{n_{in}})}.
          \end{equation}
          In Figure \ref{fig:plotboostingaw}, we have plotted the curve of $\mathcal{R}$ as a function of $w$ for the  graph of Figure \ref{fig:isolated} and Eq.(\ref{eqadjmatrixexboosting}). 
          
         \begin{figure}[h]
              \centering
              \includegraphics[scale=0.3]{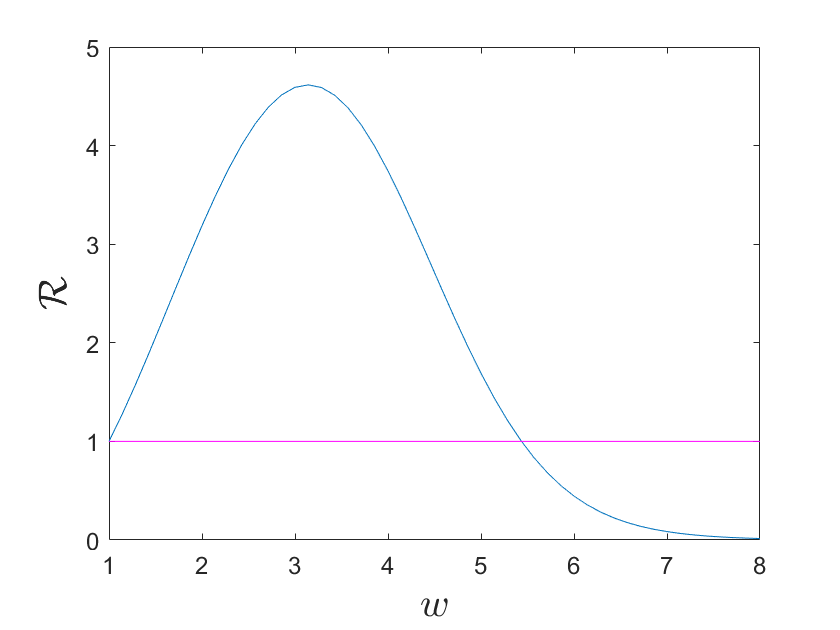}
              \caption{$\mathcal{R}$  (blue curve) as a function of $w$ for the graph of Eq.(\ref{eqadjmatrixexboosting}). Any value of $\mathcal{R}>1$ (above the horizontal purple line) indicates boosting.}
              \label{fig:plotboostingaw}
          \end{figure}
           As can be seen in Figure \ref{fig:plotboostingaw}, we can boost the probability $p(\mathbf{n}|\mathbf{n_{in}})$ up to $\approx 4.5$ times its value by using our boosting technique on the graph of Eq.(\ref{eqadjmatrixexboosting}). However, note that the boosting is not indefinite, as there is a value $w_0$ of $w$ beyond which using our technique results in lower probabilities (in Figure \ref{fig:plotboostingaw}, $w_0 \approx 5.5$). For a given fixed $n$ this behaviour is to be expected. Indeed, by looking at the upper and lower bounds of $\sigma_{max}(A_w)$ in lemma \ref{lemboundsAw} for $w\gg1$, it can be seen that these both increase linearly with the increase in $w$, and so $\sigma_{max}(A_w) \approx O(w)$. Therefore, for fixed $n>1$ , we get something like 
           \begin{equation*}
               \frac{\sigma_{max}(A)}{\sigma_{max}(A_w)} \approx O\left(\frac{1}{w}\right) << \frac{1}{w^{\frac{1}{n}}},
           \end{equation*}
           meaning that the condition in lemma \ref{lem2boosting} is violated, and consequently no boosting is anymore possible.
           
           It is interesting to speculate whether the apparent impossibility of indefinite boosting sheds light on the \emph{fundamental} incapability of quantum devices to efficiently solve $\sharp$P-hard problems, namely in this case exactly computing the permanent of an $n \times n$ matrix \cite{aaronson2011linear,valiant1979complexity}. Unfortunately, we have not been able to advance in addressing this fascinating question. 
          
          \subsection*{Second method for boosting}
          Our second technique for boosting is to boost by considering the \emph{modified} adjacency matrix
          \begin{equation*}
               \tilde{A}_{\varepsilon}=A+\varepsilon \mathbb{I}_{n \times n},
          \end{equation*}
          where $\varepsilon \in \mathbb{R}^{+}$. We will consider matrices $A \in \mathcal{M}_n(\mathbb{R}^{+})$ with non-negative entries. In this case, we have
         \begin{equation*}
            \mathsf{Per}( \tilde{A}_{\varepsilon}) \geq \mathsf{Per}(A).
          \end{equation*}
          Furthermore, $\mathsf{Per}(A)$ can be recovered from computing $\tilde{A}_{\varepsilon}$ at $n+1$ values of $\varepsilon$, then deducing $\mathsf{Per}(A)$, as is done for permanental polynomials in Section \ref{app:permpoly}.
          Let 
          \begin{equation}
          \label{eqpeps}
              p_{\varepsilon}(\mathbf{n}|\mathbf{n_{in}}):=\frac{|\mathsf{Per}(\tilde{A}_{\varepsilon})|^2}{\sigma^{2n}_{max}(\tilde{A}_{\varepsilon})}.
          \end{equation}
          What is remarkable about $p_{\varepsilon}(\mathbf{n}|\mathbf{n_{in}})$ is that, for fixed $n$, it can be made arbitrarily close to one, with increasing $\varepsilon$. To see this, consider the case where $\varepsilon\gg\mathsf{max}_{i,j}(a_{ij})$, where  $\mathsf{max}_{i,j}(a_{ij})$ is the maximum entry of $A$. In this case we have
          \begin{equation*}
              \mathsf{Per}(\tilde{A}_{\varepsilon}) \approx \mathsf{Per}(\varepsilon \mathbb{I}_{n \times n})=\varepsilon^n.
          \end{equation*}
          Also,
          \begin{equation*}
              \sigma_{max}(\tilde{A}_{\varepsilon}) \approx \sigma_{max}(\varepsilon \mathbb{I}_{n \times n})=\varepsilon.
              \end{equation*}
              Plugging these into Eq.(\ref{eqpeps}) gives
              \begin{equation*}
                  p_{\varepsilon}(\mathbf{n}|\mathbf{n_{in}}) \approx 1.
              \end{equation*}

             At this point, one is tempted to say that the boosting provided by this method is indefinite. This is a misleading conclusion however, but the reason why is subtle. In the rest of this section, we will aim to expose this subtlety, and understand under what conditions this technique provides a useful boosting.

             First, we need to prove
             \begin{lemma}
             \label{eqidea2lem1}
             $\sqrt{\sigma^2_{min}(A)+\frac{2 \varepsilon \mathsf{Trace}(A)}{n}+ \varepsilon^2} \leq \sigma_{max}(\tilde{A}_{\varepsilon}) \leq \sqrt{n}(\norm{A}_{\infty}+\varepsilon)$, where $\sigma_{min}(A)$ is the lowest singular value of $A$.
             \end{lemma}
             \begin{proof}
              To compute the upper bound, recall the identity
              $\sigma_{max}(\tilde{A}_{\varepsilon}) \leq \sqrt{n}\norm{\tilde{A}_{\varepsilon}}_{\infty}$.
              Now, $\norm{\tilde{A}_{\varepsilon}}_{\infty}=\norm{A+\varepsilon \mathbb{I}_{n \times n}}_{\infty} \leq \norm{A}_{\infty}+\varepsilon\norm{\mathbb{I}_{n \times n}}_{\infty}=\norm{A}_{\infty}+\varepsilon$, where the last part of this equation follows from applying the triangle inequality for norms. Plugging this into the above identity completes the proof for the upper bound.
              
              For the lower bound, consider
              \begin{multline}
              \label{eqidea21}
        \mathsf{Trace}(\tilde{A}_{\varepsilon}\tilde{A}^T_{\varepsilon})=\mathsf{Trace} \big( (A+\varepsilon \mathbb{I}_{n \times n})(A^T+\varepsilon \mathbb{I}_{n \times n})
                  \big)=\\ \mathsf{Trace}(AA^T)+2\varepsilon \mathsf{Trace}(A)+ \varepsilon^2\mathsf{Trace}(\mathbb{I}_{n \times n})=\\ \mathsf{Trace}(AA^T)+2\varepsilon \mathsf{Trace}(A)+ n\varepsilon^2=\\ \sum_{i}\sigma^2_{i}(A)+2\varepsilon \mathsf{Trace}(A)+ n\varepsilon^2.
               \end{multline}
         Now,
         \begin{equation*}
             \mathsf{Trace}(\tilde{A}_{\varepsilon}\tilde{A}^T_{\varepsilon})=\sum_{i}\sigma^2_{i}(\tilde{A}_{\varepsilon}) \leq n \sigma^2_{max}(\tilde{A}_{\varepsilon}),
         \end{equation*}
         and
         \begin{equation*}
             \sum_{i}\sigma^2_{i}(A)+2\varepsilon \mathsf{Trace}(A)+ n\varepsilon^2 \geq n \sigma^2_{min}(A)+2\varepsilon \mathsf{Trace}(A)+ n\varepsilon^2.
         \end{equation*}
         Plugging these into Eq.(\ref{eqidea21}) gives
         \begin{equation}
             \label{eqlowboundaepsilon}
             n \sigma^2_{max}(\tilde{A}_{\varepsilon}) \geq n \sigma^2_{min}(A)+2\varepsilon \mathsf{Trace}(A)+ n\varepsilon^2.
         \end{equation}
         Dividing both sides of Eq.(\ref{eqlowboundaepsilon}) by $n$ then taking the square root results in the desired lower bound. This concludes the proof of lemma \ref{eqidea2lem1}.
          \end{proof}
          
          Recall  that we can write 
          \begin{equation*}
              \mathsf{Per}(\tilde{A}_{\varepsilon})=\sum_{i=0,\dots,n}c_{i}\varepsilon^i,
          \end{equation*}
             with $c_{0}=\mathsf{Per}(A)$, and $c_{n}=\mathsf{Per}(\mathbb{I}_{n \times n})=1$, and $c_i \geq 0$ since they are related to sums of permanents submatrices of $A$ \cite{merris1981permanental}. Let $\lambda_1:=\mathsf{max}(c_0,c_1,\dots,c_n)$ and $\lambda_2:=\mathsf{min}(c_0,c_1,\dots,c_n)$ be the maximum and minimum values of the coefficients $c_i$. We now prove that
             \begin{lemma}
             \label{lemboundsperaepsilon}
             $ \lambda_2 \frac{\varepsilon^{n+1}-1}{\varepsilon -1 }\leq \mathsf{Per}(\tilde{A}_{\varepsilon}) \leq \lambda_1 \frac{\varepsilon^{n+1}-1}{\varepsilon -1 }$.
             \end{lemma}
             \begin{proof}
             The proof of the upper bound follows first from noting that
             $\sum_{i=0,\dots,n}c_i\varepsilon^i\leq \lambda_1 \sum_{i=0,\dots,n}\varepsilon^i$, then by using the geometric series identity $\sum_{i=0,\dots,n}\varepsilon^i=\frac{\varepsilon^{n+1}-1}{\varepsilon -1 }$. The proof of the lower bound is similar, but the starting point is $\sum_{i=0,\dots,n}c_i\varepsilon^i\geq \lambda_2 \sum_{i=0,\dots,n}\varepsilon^i$.
             \end{proof}
             We will now consider the case where $\varepsilon \to \infty$ and $n$ is fixed. In this case, lemma \ref{eqidea2lem1} implies
             \begin{equation}
             \label{eqasympsigmaaeps}
             \sigma_{max}(\tilde{A}_{\varepsilon}) \approx O(\varepsilon).
             \end{equation}
             Similarly, lemma \ref{lemboundsperaepsilon} gives
             \begin{equation}
             \label{eqasympperaepsilon}
                 \mathsf{Per}(\tilde{A}_{\varepsilon}) \approx O(\varepsilon^n).
             \end{equation}
             
             With these equations in hand, we will now argue that, after a certain point, estimating $\mathsf{Per}(A)$ starting from $\mathsf{Per}(\tilde{A}_{\varepsilon})$ will require a higher sample complexity (number of experiments needed to be performed) than estimating $\mathsf{Per}(A)$ directly. This will show why, although the probabilities $p_{\varepsilon}(\mathbf{n}|\mathbf{n_{in}})$ can be boosted indefinitely with our method, our method will cease being advantageous after a certain value of $\varepsilon$.
             
             Recall that in order to estimate probabilities in our setup to within additive error $\frac{1}{\kappa}$, we require $O(\kappa^2)$ samples from standard statistical arguments \cite{hoeffding1994probability}. In order to estimate $|\mathsf{Per}(\tilde{A}_{\varepsilon})|^2$  to a good precision in our setup, we need $\frac{1}{\kappa} \approx O(\frac{1}{\sigma^{2n}_{max}(\tilde{A}_{\varepsilon})}) \approx O(\frac{1}{\varepsilon^{2n}})$, since the output probabilities (proportional to $|\mathsf{Per}(\tilde{A}_{\varepsilon})|^2$) are scaled down by $\sigma^{2n}_{max}(\tilde{A}_{\varepsilon})$ in our setup. 
            Thus,  the total number of experiments we need to perform to estimate $\mathsf{Per}(\tilde{A}_{\varepsilon})$ (and therefore estimate from it $\mathsf{Per}(A)$) is
             \begin{equation}
                 \label{eqexperiments1}
                 \mathcal{E}_{\tilde{A}_{\varepsilon}}=O( \kappa^2)=\sigma^{4n}_{max}(\tilde{A}_{\varepsilon})=O(\varepsilon^{4n}).
             \end{equation}
             By a similar argument, directly estimating $A$ by encoding $A$ into our setup without modification and collecting samples requires
             \begin{equation}
                 \label{eqexperiment2}
                 \mathcal{E}_{A}=O(\sigma^{4n}_{max}(A))=O(1),
             \end{equation}
             since $n$ is fixed.
             It is now clear that, with $\varepsilon$ increasing, there is a point $\varepsilon_0$ after which $$\mathcal{E}_{\tilde{A}_{\varepsilon}}>\mathcal{E}_{A}.$$ At this point, it will  no longer be advantageous to use our modification to compute the permanent of $A$. 
             
             \bigskip 
             
             As a concluding remark for this section, although the methods we discussed here for boosting do not provide an indefinite advantage, they may nevertheless be useful to obtain advantages in practice, especially in the context of NISQ hardware where the number of photons  $n$ in our setup is small-to-modest.
            
             \section{Numerics}
             \label{app:numerics}

             Here we provide numerical support for Theorem \ref{thperm}, stating that the permanent of a given graph (with even number of vertices and edges) is upper bounded by a monotonically increasing function of the number of edges. This is at the heart of why  our setup can be used to identify dense subgraphs, as denser subgraphs tend to appear more when sampling. We performed numerical tests for random graphs of different number of vertices, and increasing number of edges per each vertex number. The results are plotted in Figures \ref{permex1} and \ref{permex2}. We can indeed observe, as predicted by Theorem \ref{thperm}, that the exact value of the permanent increases with the graph edge probability.
             \begin{figure}[h]
             \centering
              \includegraphics[scale=0.4]{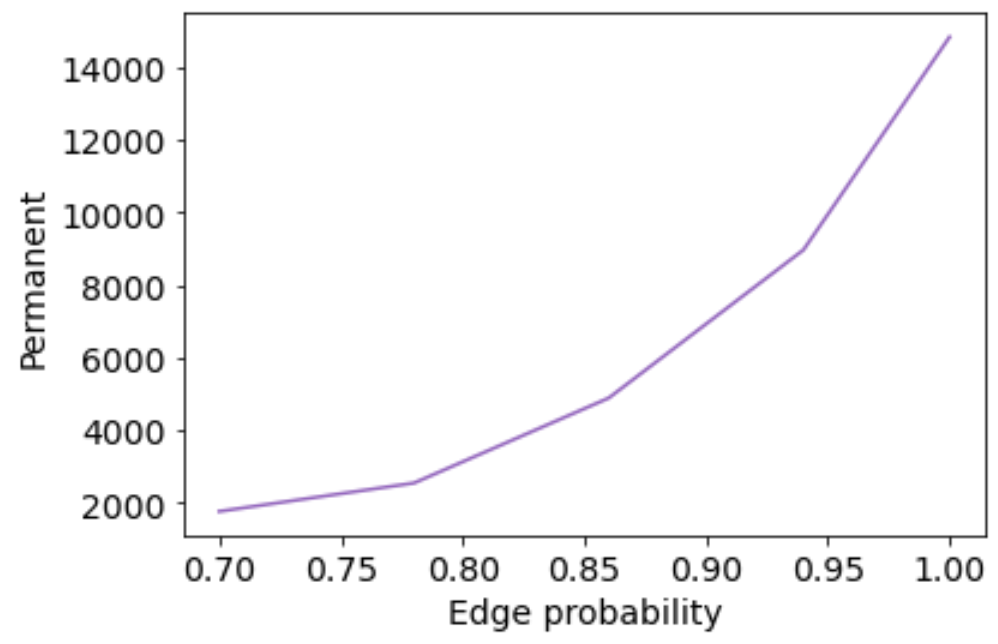}
              \caption{Mean value of the permanent of 15 randomly generated graphs of 8 vertices plotted in function of edge probability. The edge probability represents the probability that any two vertices $i$ and $j$ of the randomly generated graph are connected by an edge.}
              \label{permex1}
             \end{figure}
             \begin{figure}[h]
             \centering
              \includegraphics[scale=0.4]{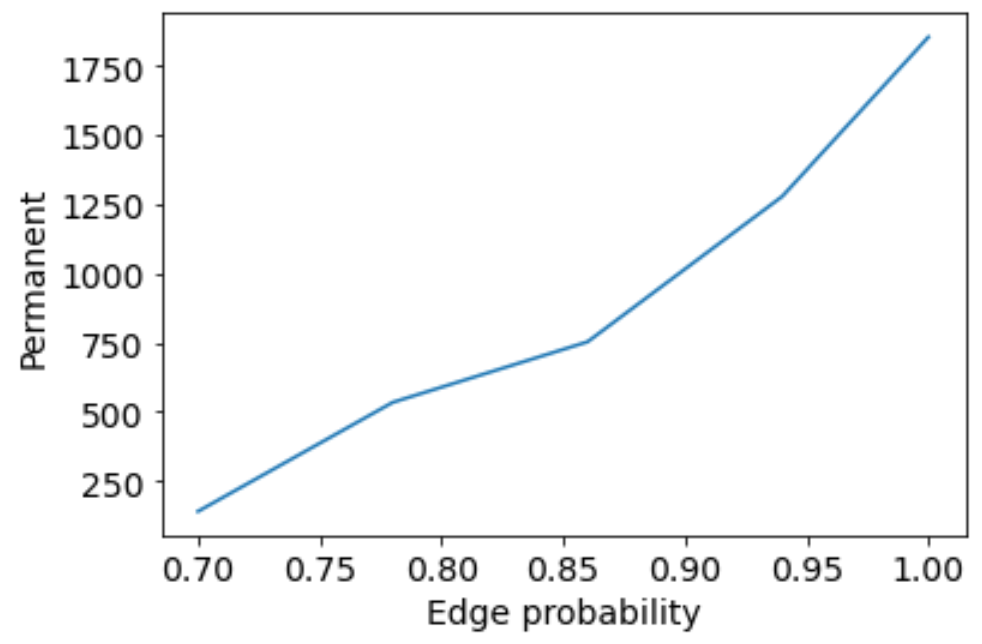}
              \caption{Mean value of the permanent of 15 randomly generated graphs of 7 vertices plotted in function of edge probability.}
              \label{permex2}
             \end{figure}

             Finally, we constructed code to test our first method for boosting (see Appendix \ref{app:boosting}).
We considered the graph of Figure \ref{fig:boosting1}, which has the following adjacency matrix
\begin{equation}
\label{eqboostnum1}
    A=\begin{pmatrix} 0& 1& 1& 1& 1& 0 \\ 1& 0& 1& 1& 1& 1 \\ 1& 1& 0& 1& 1& 0 \\1& 1& 1& 0& 1& 0 \\ 1& 1& 1& 1& 0& 0 \\ 0& 1& 0& 0& 0& 0 \end{pmatrix}.
\end{equation}

\begin{figure}[h]
    \centering
    \includegraphics[scale=0.3]{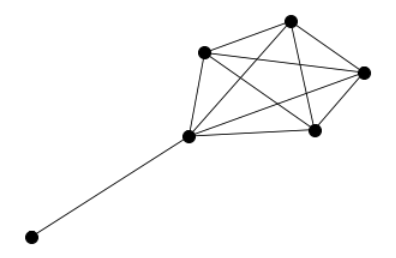}
    \caption{Test graph for boosting.}
    \label{fig:boosting1}
\end{figure}
Note that $\mathsf{Per}(A)=9$. Multiplying the last row of the matrix in Eq.(\ref{eqboostnum1}) by $w \in \{1,2,3,4,5,6\}$, we obtain a matrix
\begin{equation}
\label{eqboostnum2}
    A_w=\begin{pmatrix} 0& 1& 1& 1& 1& 0 \\ 1& 0& 1& 1& 1& 1 \\ 1& 1& 0& 1& 1& 0 \\1& 1& 1& 0& 1& 0 \\ 1& 1& 1& 1& 0& 0 \\ 0& w& 0& 0& 0& 0 \end{pmatrix}.
\end{equation}
For each value of $w \in \{1,2,3,4,5,6\}$, we computed an estimate of $\mathsf{Per}(A_w)$, and deduced from it an estimate of $\mathsf{Per}(A)$ (see Eq.(\ref{eqrelAwA})) using 100 post-selected samples, and recorded the time needed to collect those samples. We report these results in Table \ref{ta2}. As can be clearly observed in Table \ref{ta2}, we observe boosting for $w \in \{2,3,4\}$, as manifested in the time needed to compute an estimate of the permanent with these values of $w$ versus the time needed to compute it with no modification ($w=1$) of the adjacency matrix. We also observe that multiplying by $w>4$
ceases to boost the desired output probabilities.
\begin{table}[h]
\centering
\begin{tabular}{|l|l|l|}
\hline
$w$ & Permanent estimation & Time \\ \hline
1     & 8.776  & 104min 43.6s   \\ \hline%
2     & 8.694  & 35min 5.2s    \\ \hline %
3     & 8.613  & 30min 13.8s    \\\hline %
4     & 9.303  & 51min 18.5s      \\ \hline%
5     & 9.637  & 158min 27.4s   \\  \hline%
6     & -----  & $>$ 200min   \\  \hline%
\end{tabular}
\caption{Results of testing boosting for different multiplication values by multiplying by $w$ node $5$ of graph represented in Figure \ref{fig:boosting1}. The middle column of this table contains estimates of $\mathsf{Per}(A)$ for each value of $w$ tested, and the rightmost column contains the times required to compute these estimates.}
\label{ta2}
\end{table}

\end{document}